\documentclass[11pt,onecolumn,letterpaper]{IEEEtran}

\usepackage[ruled,vlined]{algorithm2e}
\usepackage{bbm}
\usepackage{cite}
\usepackage{amsfonts}
\usepackage{graphicx}
\usepackage{subfigure}
\usepackage{epsfig}
\usepackage{color}

\newtheorem{definition} {Definition}
\newtheorem{lemma} {Lemma}
\newtheorem{theorem} {Theorem}

\newtheorem{corollary}{Corollary}

\newcommand{\G}[2]{\mbox{$\mathcal{G}(#1,#2)$}} % G(n,r)
\newcommand{\g}[2]{\mbox{$\mathcal{L}(#1,#2)$}}
\newcommand{\pr}[1]{\mbox{$\mathcal{P}\left[#1\right]$}}
\newcommand{\E}[1]{\mbox{$\mathcal{E}\left[#1\right]$}}
\newcommand{\B}[3]{\mbox{$\mathcal{B}_{#3}^{#1}(#2)$}}
\newcommand{\Boxes}[3]{\mbox{$Box_{#3}^{#1}(#2)$}}
\newcommand{\ceil}[1]{\mbox{$\left\lceil #1\right\rceil$}}
\newcommand{\floor}[1]{\mbox{$\left\lfloor  #1\right\rfloor$}}

\newcommand{\R}{\mathbb{R}}

\newcommand{\oft}[1]{\mbox{$#1^{(t)}$}}
\newcommand{\void}[1]{\mbox{$\mathcal{V}_{#1}$}}
\newcommand{\Gn}{\mbox{$\mathcal{G}_n$ }}
\newcommand{\Ln}{\mbox{$\mathcal{L}_n$ }}
\newcommand{\Hn}{\mbox{$\mathcal{H}_n$ }}
\begin{document}

\title{Hierarchical Routing over Dynamic Wireless Networks}
\author{
\authorblockN{Dominique Tschopp, Suhas Diggavi, and Matthias Grossglauser}
\authorblockA{\\School of Computer and Communication Sciences\\
Ecole Polytechnique F\'ed\'erale de Lausanne (EPFL)\\
1015 Lausanne, Switzerland\\
Email: Dominique.Tschopp@epfl.ch, Suhas.Diggavi@epfl.ch, Matthias.Grossglauser@epfl.ch
}
}
\maketitle
\begin{abstract}

Wireless network topologies change over time and maintaining routes requires frequent
updates. Updates are costly in terms of consuming throughput
available for data transmission, which is precious in
wireless networks. In this paper, we ask whether there
exist low-overhead schemes that
produce low-stretch routes. This is studied by using the underlying
geometric properties of the connectivity graph in wireless networks.
%For a class of models for mobile wireless network that fulfill some
%mild conditions on the connectivity and on mobility over the time of
%interest, we can design distributed routing algorithm that maintains
%the routes over a changing topology. This scheme needs only node
%identities and therefore integrates location service along with
%routing, therefore accounting for the complete overhead.  We analyze
%the worst-case (conservative) overhead and route-quality (stretch)
%performance of this algorithm for the aforementioned class of wireless
%network connectivity and mobility models.  In particular for these
%models, we show that our algorithm allows constant stretch routing
%with a network wide control traffic overhead of $O(n\log^2 n)$ bits
%per mobility time step (time-scale of topology change) translating to
%$O(\log^2 n)$ overhead per node (with high probability for wireless
%networks with such mobility model). Additionally, we can
%reduce the maximum overhead per node by using a load-balancing
%technique at the cost of a slightly higher average overhead.  We also
%demonstrate through numerics that these worst-case bounds are quite
%conservative in terms of the constants derived theoretically.

\end{abstract}

%%A category including the fourth, optional field follows...
%\category{D.2.8}{Software Engineering}{Metrics}[complexity measures, performance measures]
%
%\terms{Delphi theory}
%
% NOT required for Proceedings

\newpage
\section{Introduction}
\label{sec:intro}
A major challenge in the design of wireless ad hoc networks is the
need for distributed routing algorithms that consume a minimal amount
of network resources. This is particularly important in dynamic
networks, where the topology can change over time, and therefore
routing tables must be updated frequently. Such updates incur control
traffic, which consumes bandwidth and power.  It is natural to ask
whether there exist low-overhead schemes for dynamic wireless networks
that could produce and maintain efficient routes. In this paper we
consider dynamically changing connectivity graphs that arise in
wireless networks. Our performance metric for the algorithms is the
average signaling overhead incurred over a long time-scale when the
topology changes continuously. We design a routing algorithm which can
cope with such variations in topology. We maintain efficient routes
from any source to any destination node, for each
instantiation\footnote{We assume inherently that the round-trip time
(RTT) of a packet from source to destination is much smaller than the
time-scale of topology change.}  of the connectivity graph. By
efficient, we mean that we want to guarantee that the route is within
a (small) constant factor, called {\em stretch} of the shortest path
length. In order to route to a destination, we need only the identity
of the destination and not its address {\em i.e., } the control traffic
to maintain the mapping between node identity and address/location is
incorporated into the overhead. Therefore, in the wireless routing
terminology, we have included the ``location service'' in the control
signaling requirement, and therefore hope to characterize the complete
overhead needed to maintain efficient routes.

In order to develop and analyze the routing algorithms we utilize the
underlying geometric properties of the connectivity graphs which arise
in wireless networks. This geometric property is captured by the {\em
doubling dimension} of the connectivity graph.  A graph induces a
metric space by considering the shortest path distance between nodes
as the metric distance.  The doubling dimension of a metric space is
the number of balls of radius $R$ needed to cover a ball of radius
$2R$. For example a Euclidean space has a low doubling dimension as
will be illustrated in Section \ref{sec:mod}. A metric space having a
low (constant independent of the cardinality of the metric space)
doubling dimension is called ``doubling''. We show that several
wireless network graphs (under conditions given in Section
\ref{sec:mod}) are doubling and therefore enable the design and
analysis of hierarchical routing strategies. In particular, it is not
necessary to have uniformly distributed nodes with geometric
connectivity for the doubling property to hold, as illustrated in
Figure \ref{fig:wall} in Section
\ref{sec:mod}.  Therefore, the doubling property has the potential to
enable us to design and analyze algorithms for a general class of
wireless networks.  Moreover, for a large class of mobility models,
the sequence of graphs arising due to topology changes are all
doubling (for specific wireless network models). Since there are only
``local'' connectivity changes due to mobility, there is a smooth
transition between these doubling graphs. We can utilize the locality
of topology changes to develop lazy updates methods to reduce
signaling overhead.

We show that several important wireless network models produce
connectivity graphs that are doubling. In particular, we show that the
geometric random graph with connectivity radius growing as $\sqrt{\log
n}$ with network size $n$; the fully connected regime of the dense or
extended wireless network with signal-to-interference-plus-noise ratio
(SINR) threshold connectivity; some examples of networks with
obstacles and non-homogeneous node distribution.  We define a sequence
of wireless connectivity graphs to be {\em smooth} if each of the
graphs is doubling and the shortest path distance between two nodes in
the graph changes smoothly (defined in Section \ref{sec:mod}). These
for mild regularity conditions on the mobility model.

Our main results in this paper are the following. (i) For smooth
geometric sequence of connectivity graphs, we develop a routing
strategy based on a hierarchical set of beacons with scoped
flooding. We also maintain cluster membership for these beacons in a
lazy manner adapted to the mobility model and doubling dimension. (ii)
We develop a worst-case analysis of the routing algorithm in terms of
total routing overhead and route quality (stretch). We show that we
can maintain constant stretch routes while having an average
network-wide traffic overhead of $O(n\log^2 n)$ bits per mobility time
step. The load-balanced algorithm would require $O(\log^3 n)$ bits per
node, per mobility time. Through numerics we show that the
theoretically obtained worst-case constants are conservative.

\subsection{Related Work}

Routing in wireless networks has been a rich area of enquiry over the
past decade or more. The two main paradigms for routing have been
geographic routing and topology based routing.  
Geographic routing (see for instance \cite{GPSR} and references therein) exploits
the inherent geometry of wireless networks, and bases routing decisions
directly on the Euclidean coordinates of nodes.
Their performance depends on how well the Euclidean coordinate system captures
the actual connectivity graph, and these approaches 
can therefore fail in the presence of node or channel inhomogeneity
(like in Figure \ref{fig:wall} in Section \ref{sec:mod}). Another important, but often overlooked, issue with geo-routing is that 
geographical positions of the nodes need to be stored and continuously updated in a distributed database in the network,
to allow sources of messages to determine the current position of the destination.
This database is called a \emph{location service} (see for instance \cite{GLS}) 
and must be regularly updated so that source nodes can query it. 
Location services typically rely on some a-priori knowledge 
of the geographical boundaries of the network.
This is necessary because these approaches typically establish a correspondence
(for example, through a hash function)
between a node identifier and one or several geographical locations where location
information about that node is maintained. 
An important feature of our work is that we consider the total overhead incurred
by the update and lookup operations of the location service, and the overhead
of the routing algorithm itself.

Topology based routing schemes (see \cite{Perkins97} and \cite{DSR}) 
do not utilize the underlying geometry of wireless
connectivity graphs, but instead compute routes based directly on that graph.
To reduce overhead, most of these schemes only establish routes {\em on demand} through
a route discovery operation, rather than continuously maintaining a route between
every pair of destinations; in this respect, they differ significantly from their
counterparts for the wired Internet (such as OSPF, IS-IS, and RIP).
Recently established routes are cached in order to allow their reuse by future messages.
In distance-vector based approaches (e.g., \cite{Perkins97}, this cached state resides
in the intermediate nodes that are part of a route, whereas in source-routing approaches
(e.g., \cite{DSR}), the cached state resides in the source of a route.
Despite such optimizations, topology-based approaches suffer from the large overhead
of frequent route discovery operations in large and dynamic networks.
This issue was, in fact, the reason why geo-routing approaches have reached prominence.
%It also has significant overhead requirements to maintain routes over a changing
%topology.  
%Routing was perhaps first proposed by Kleinrock in his
%seminal work \cite{Kleinrock77} for static wired networks.
%comment matt: that paper was really about hierarchical routing, not really relevant here i think

Two schemes that utilize the underlying geometry of graphs in {\em
static} wireless networks algorithms are the works presented in
\cite{Mobicom03} and the beacon vector routing (BVR) introduced in
\cite{BVR}. Both these schemes are heuristics which build a virtual
coordinate system over which routing takes place. They were shown to
work well through numerics. However, they utilize an external
addressing scheme to make a correspondence between addresses and
names. In \cite{tdgwINFOCOM}, routing on dynamic networks using a
virtual coordinate system was studied. For large scale dynamic
wireless networks, these heuristics pointed to significant advantages
to using some geometric properties for routing and addressing. These
results motivated the questions studied in this paper.
%comment matt: i don't fully understand this paragraphy, let's discuss

There has been a vast amount of theoretical research on efficient
routing schemes in wired ({\em i.e.,} static) networks (see for example
\cite{Gavoille01}). Most of this work has been focused on the
trade-off of memory (routing table size) and routing stretch. There
are two main variants of such routing schemes (i) {\em labeled} (or
{\em addressed}) routing schemes, where the nodes can be assigned
addresses so as to reflect topological information; (ii) {\em named}
routing, where nodes have arbitrary names, and as part of the routing,
the location (or address) of the destination needs to be obtained
(similar to a location service).  This examines the important question
of how the node addresses need to be published in the network.
Routing in graphs with finite doubling dimension has been of recent
interest (see
\cite{Konjevod2006}, and references therein). In particular
\cite{Talwar2004} showed that one could get constant stretch routing with
small routing table sizes for doubling metric spaces, when we use
labeled routing. This result was improved to make routing table sizes
smaller in \cite{Chan2005}. The problem of named routing over graphs with
small doubling dimension has been studied in \cite{Konjevod2006} and
\cite{Abraham2006}, and references therein. To the best of our knowledge,
there has been no prior work on {\em dynamic} graphs over doubling
metric spaces and on control traffic overhead. It is worth
pointing out that there is no direct correspondence between control
traffic and memory. Bounds on memory do not take into account the
amount of information which needs to be sent around in the network in
order to build routing tables. A good illustration is the computation
of the shortest path between two nodes $u$ and $v$ in a graph. While
it is sufficient for every node on the path between these two nodes to
have one entry for $v$ (of roughly $\log n$ bits \emph{i.e.,} the name
of the next hop), computing that shortest path requires a breadth
first search of the communication graph and leads to a control traffic
overhead of $O(n \log n)$ bits.

% Kleinrock hierarchical routing not for dynamic
% Compact routing for memory not for overhead
% Tree routing with random partitions - Talwar, SODA paper
% Unlabelled routing Kanjevod etc for doubling metrics.
% Dynamic networks Infocom 07 formulation, Virtual co-ordinates (Rao paper), Beacon vector routing.
% Geographic routing - GPRS/Wattenhoffer, bad stretch? 
% Grid location service gives node coordinates/location service, geo-routing pitfalls.
% classical adhoc routing (DSR/AODV) failure in dynamic networks.

% Differences - dynamic, without co-ordinates, no separate location service, addresses/locations
% depend on connectivity not on GPS. Provable bounds for general class.

\section{Models and Definitions}
\label{sec:mod}

A wireless network consists of a set of $n$ nodes spread across a
geographic area in the two-dimensional plane. We model the network
region as the square area $\left[0,\sqrt{n})\right.\times
\left[0,\sqrt{n})\right.$. The $n$ nodes move randomly in this area
and we denote by $x^{(t)}(u)$ the position of node $u$ at time $t$.
The connectivity between two nodes is represented by an edge on the
connectivity graph $\mathcal{G}_n^{(t)}$ if they can communicate
\emph{directly} over the wireless channel.  The connectivity between
two nodes depends on the distance between the two nodes (and could
also depend on the presence of other nodes, see Section
\ref{subsec:SINRmodel}). We consider that when a node $u$ transmits on
the wireless channel, it broadcasts to all its neighbors in the
connectivity graph $\mathcal{G}_n^{(t)}$. Consequently, one
transmission of a packet is sufficient for all direct neighbors to
receive that packet.  To make the notation lighter, we will only add
the dependence on time if it is necessary to avoid confusion. The
distance $d^{(t)}(u,v)$ between nodes $u$ and $v$ is the shortest path
distance between these nodes in $\mathcal{G}_n^{(t)}$. Note that
$d(.,.)$ is a metric on $\mathcal{G}_n^{(t)}$, \emph{i.e.,} the
distance between a node and itself is zero, the distance function is
symmetric and the triangle inequality applies. We will now define a
\emph{ball} of radius $R$ around a node $u$. It is simply the set of
nodes within distance $R$ of $u$. More formally, we can define it more
generally for any metric space as follows:
\begin{definition}
A \textit{Ball} $\B{(t)}{u}{R}$ around node $u$ at time $t$ in a
metric space $\mathcal{X}$ is the set $\left\{v\in
\mathcal{X}|d^{(t)}(u,v)\leq R\right\}$.
\end{definition}
In order to bound the control traffic overhead, we will recursively
subdivide the connectivity graph into balls. It will be crucial for us
to bound the number of balls of radius $R$ necessary to cover a ball
of radius $2R$ around some node $u$. In other words, we want to find
the smallest number of nodes $v_i$ such that all nodes within $2R$ of
$u$ are also within $R$ of some node $v_i$. The notion of doubling
dimension of a metric space captures this idea.
\begin{definition}
The \textit{doubling dimension} of a metric space $\mathcal{X}$ is the
smallest $\alpha$ such that any ball of radius $2R$ can be covered by
at most $\alpha$ balls of radius $R$, for all $R\geq
\min_{(u,v)}d(u,v)$\emph{ i.e.,} $\forall u\in X\ \exists\
S_u\subseteq\mathcal{X}, |S_u|\leq \alpha \mbox{ and }$
\[ 
\B{(t)}{u}{2R}\subseteq\bigcup_{j\in S_u}\B{(t)}{j}{R}
\]
\end{definition}
Moreover, if $\alpha$ is a constant, we have the following definition:
\begin{definition}[Doubling metric space]
We a metric space $\mathcal{X}$ is \emph{doubling} if its doubling dimension is a constant.
\end{definition}
A good way to illustrate and understand the concept of doubling
dimension and doubling metric space is to look at the metric space
defined by a set of points $\mathcal{X}$ in $\R^2$ with the Euclidean
distance. A ball of radius $2R$ around a point $x$ will simply be a
disc of radius $2R$ around this point. To cover this disc, we will
select a set of points such that all the surface is covered by the
corresponding set of discs of radius $R$. Note that the number of
discs required will not depend on R, and consequently this metric
space would be \emph{doubling} (see Figure \ref{fig:deu}).  Further, a
metric space is said to be \emph{doubling} if its doubling dimension
is a constant, independent of the number of nodes $n$.
\begin{figure}[htbp]
\centering
\includegraphics[width=0.6\textwidth,keepaspectratio]{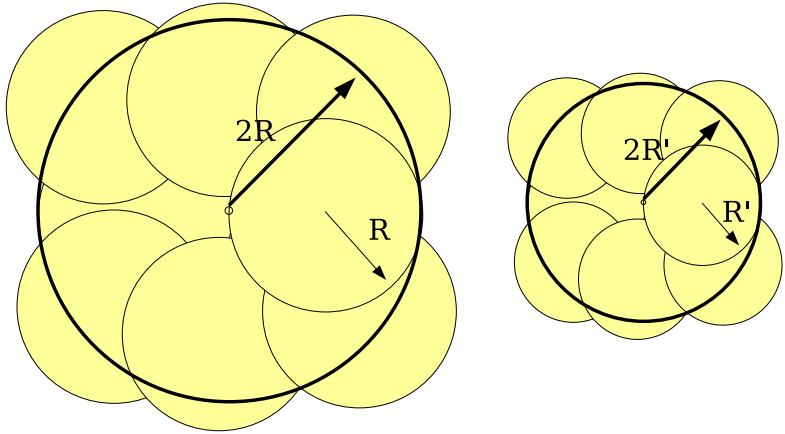}
\caption{The metric space defined by a set of points in $\R^2$ and the
Euclidean distance is doubling. Indeed, we can cover a disc of radius
$2R$ by a constant ($8$ in this case) number of discs of radius $R$,
whatever the value of $R$.}
\label{fig:deu}
\end{figure}   

In Section \ref{subsec:GnR}, we describe the geometric random graph model,
which will be the canonical model we will use to illustrate the ideas
of the paper.  We also give an example of a non-homogeneous network to
which our results can be applied. In Section \ref{subsec:SINRmodel},
we will develop the model where connectivity is determined by the
SINR, and we have uniform transmit power and full connectivity. We
give the requirements for the mobility model to result in a {\em
smooth} sequence of wireless network graphs in Section
\ref{subsec:MobSmooth}. We state the underlying assumptions and
give a table of notations in Section \ref{sub:Ass}.

\subsection{Geometric random graph}
\label{subsec:GnR}

We denote the geometric random graph by $\G{n}{r_n}$ and define its
connectivity as follows.
\begin{definition}
\label{def:GnR}
A random geometric graph $\G{n}{r_n}$ has an unweighted edge between
nodes $u$ and $v$ if and only if $||x(u)-x(v)||<r_n$, where $\{x(u)\}$
are chosen independently and uniformly in
$\left[0,\sqrt{n})\right.\times \left[0,\sqrt{n})\right.$.
\end{definition} 
In this paper we will be interested in fully connected geometric
random graphs, and therefore focus on the case $r_n>\sqrt{\log n}$
\cite{GK98}.  As a natural extension, we can also define a sequence of
random graphs $\mathcal{G}^{(t)}(n,r_n)$ with an unweighted edges
between $u$ and $v$ at time $t$ if
$||x^{(t)}(u)-x^{(t)}(v)||<r_n$. Whether each graph in the sequence
$\mathcal{G}^{(t)}(n,r_n)$ corresponds to a random geometric graph as
in Definition \ref{def:GnR}, depends on the mobility model for the
nodes. We discuss this in more detail in Section
\ref{subsec:MobSmooth}.

In Figure \ref{fig:wall}, we illustrate a non-homogeneous random
network where connectivity is not completely geometric as in
Definition \ref{def:GnR}. An obstacle prevents communication between
neighboring nodes, and therefore illustrates the complexities of
wireless network connectivity. This example is revisited in Section
\ref{sec:npp}, where we show that though this connectivity graph is
more complicated than $\G{n}{r_n}$, it is still doubling, and
therefore the algorithms developed in this paper are applicable. This
illustrates the advantage of our approach to network modeling.

\begin{figure}[htbp]
\centering
\includegraphics[width=0.3\textwidth,keepaspectratio]{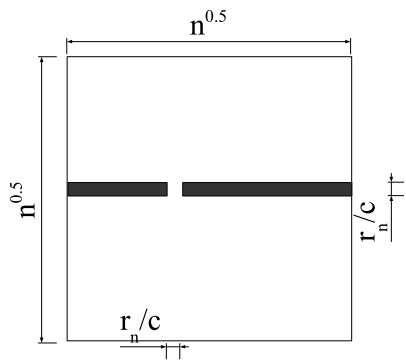}
\caption{$n$ nodes are distributed uniformly at random on a square
area of side $\sqrt{n}$. A wall of width $r_n/c$ is added, which only
has a small hole in the middle. Again, we assume $r_n>\sqrt{\log
n}$. Nodes cannot communicate through the wall. Finally, we remove the
nodes below the wall, which leads to an inhomogeneous node
distribution.}
\label{fig:wall}
\end{figure}   
\subsection{SINR full connectivity}
\label{subsec:SINRmodel}

Since the wireless channel is a shared medium, the transmissions between nodes
interfere with each other. However, the signal strength decays as a function of
the distance traveled, and therefore we can define the SINR for transmission
from node $u$ to $v$ as,
\begin{equation}
\label{eq:SINR}
\mathrm{SINR} = \frac{P_n||x(u)-x(v)||^{-\beta}}{N_0+\sum_{w\neq
u,v}P||x(w)-x(v)||^{-\beta}},
\end{equation}
where $\beta$ is a distance loss (decay) parameter depending on the
propagation environment, $P_n$ is the common transmit power of the nodes
and $N_0$ is the noise power. We can of course easily adapt this to
have power control for the nodes. A transmission is successful if the SINR is above some constant threshold value $\varsigma$. For static nodes, just as in the
case of geometric random graph, we assume that the node locations
$\{x(u)\}$ are chosen independently and uniformly in
$\left[0,\sqrt{n})\right.\times \left[0,\sqrt{n})\right.$. This model for wireless networks has been extensively studied in the literature (see \cite{GK00,Kulkarni2002}). The authors base their analysis of the capacity of wireless networks on a TDMA scheme for the SINR connectivity
model of (\ref{eq:SINR}). We argue here that the structure of the resulting
connectivity graph is identical to that of $\G{n}{r_n}$, for
$r_n>\sqrt{\log n}$. Therefore, the results we prove for $\G{n}{r_n}$,
would also be applicable to such graphs. In practice, it is a non-trivial task to design a distributed scheduling protocol (MAC layer protocol) that mimics the behavior of this TDMA scheduler. However, these MAC layer implementation issues are far beyond the scope of this document (see for instance \cite{sigshah}). We only make the argument here that the connectivity graph resulting from such a TDMA scheme would yield the same behavior as a $\G{n}{r_n}$.  
\par 
We will subdivide the network into small squares of side $s_n=\frac{r_n}{c}$. We need to show that if two nodes $u$ and $v$ are in neighboring small squares (and so have the guarantee that they can communicate under the $\G{n}{r_n}$ model as we will see in the sequel), then there exists a TDMA scheme that allows them to communicate under the SINR connectivity
model of (\ref{eq:SINR}). If this is the case, then we can apply the same proof techniques for both models. We let the maximum transmission power grow in the same way as we did for the $\G{n}{r_n}$ model\footnote{Note that the $\G{n}{r_n}$ model corresponds to the SNIR model without interferences. Indeed, if we remove interferences, two nodes can communicate whenever $\frac{P_n ||x(u)-x(v)||^{-\beta}}{N_0}>\varsigma$ for some threshold value $\varsigma$. Hence, two nodes can communicate whenever $||x(u)-x(v)||<(\frac{P_n}{N_o \varsigma})^{1/\beta}$. In particular, we let $P_n=(N_o\varsigma r_n)^\beta$.} \emph{i.e.} $P_n\leq(N_o\varsigma r_n)^\beta $. Additionally, we want to design a TDMA scheme such that the capacity of all links is at least $O(\frac{1}{\log n})$ $\left[bits/sec\right]$. It can be shown (see \cite{ballsbins}) that every small square contains at most $O(\log n)$ nodes. Hence, we ask that the traffic can flow at constant rate independent of $n$ between neighboring small squares, and that each node is treated equally. Note that this requirement is very similar to the scheme proposed in \cite{GK00} in which one node per small square can transmit at constant rate to any neighboring square\footnote{The throughput achieved by this scheme is $O(\frac{1}{\sqrt{n \log n}})$ $\left[bits/second/node\right]$ when $n$ source destination pairs are chosen uniformly at random.}. 
\begin{theorem}
There exists a TDMA scheme such that all nodes can communicate with any node located in a neighboring small square at a rate of $O(\frac{1}{\log(n)})$ $\left[bits/sec\right]$. Hence, the aggregate traffic can flow between neighboring small squares at a constant rate independent of $n$. 
\end{theorem}
\begin{proof}
We take a coordinate system, and label each square with two integer coordinates. Then we take an integer $k$,
and consider the subset of squares whose two coordinates are a multiple of $k$ (see Figure \ref{fig:TDMA}). By translation, we can construct $k^2$ disjoint equivalent subsets. This allows us
to build the following TDMA scheme: we define $k^2$ time slots, during which all nodes from a
particular subset are allowed to transmit for the same duration of $O(\frac{1}{\log n})$ seconds. Each small square contains at least one and at most $O(\log n)$ nodes w.h.p. (see \cite{ballsbins} and the proof of Theorem \ref{lem:seqcg}). We assume also that at most one node per square transmits at
the same time, and that they all transmit with the same power $P_n$.
\begin{figure}[htbp]
\centering
\includegraphics[width=0.8\columnwidth,keepaspectratio]{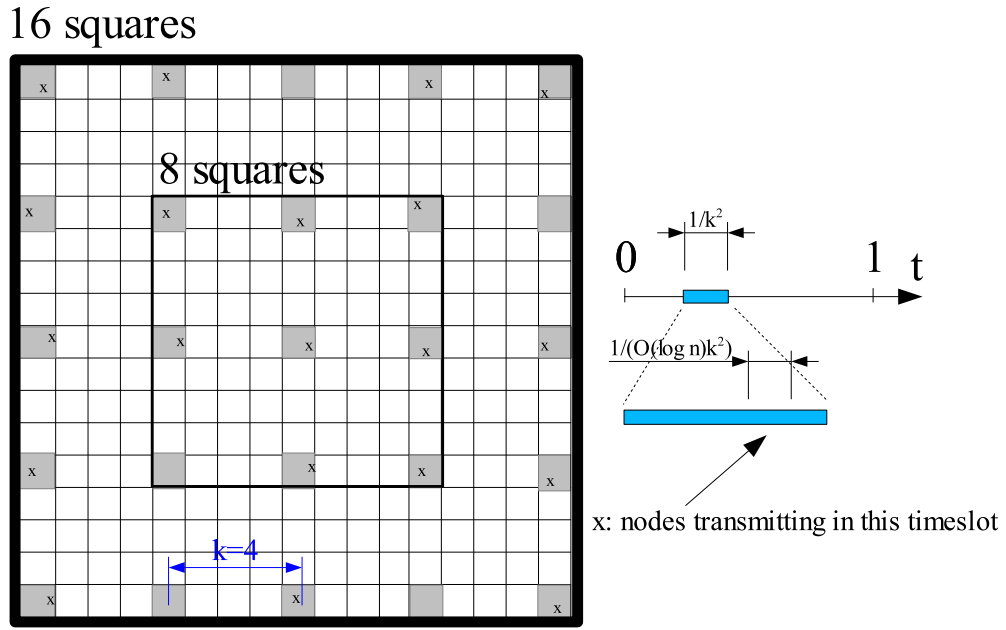}
\caption{Illustration of the TDMA scheduling scheme.}
\label{fig:TDMA}
\end{figure}   
Let us consider one particular square. We suppose that the transmitter in this square transmits towards a
destination located in a square at distance at most $1$. We compute the signal-to-interference ratio
at the receiver. First, we choose the number of time slots $k^2$ as follows: $k = 4$.
To find an upper bound to the interferences, we observe that with this choice, the transmitters in the
$8$ first closest squares are located at a distance at least $3$ (in small squares) from the receiver (see
left-hand side of Figure \ref{fig:TDMA}). This means that the Euclidean distance between the receiver and
the $8$ closest interferers is at least $2 s_n$. The $16$ next closest squares are at distance at least
$7$ (in small squares), and the Euclidean distance between the receiver and the $16$ next interferers
is therefore at least $6s_n$, and so on. The sum of the interferences $I$ can be bounded as follows:
\[
\begin{array}{ll}
I&\leq \sum_{i=1}^{\infty}8i P_n\left[2 s_n (2i-1)\right]^{-\beta}\\
&= P_n \left[2 s_n\right]^{-\beta} \sum_{i=1}^{\infty}8i \left[(2i-1)\right]^{-\beta}\\
&= (N_o\varsigma r_n)^{\beta} \left[2 \frac{r_n}{c}\right]^{-\beta} \sum_{i=1}^{\infty}8i \left[(2i-1)\right]^{-\beta}
\end{array}
\]
This term clearly converges if $\beta>2$. Now we want to bound from below the strength of the signal received from the transmitter. We observe first that the distance between the transmitter and the receiver is at most $\sqrt{2(s_n^2)}\leq 2s_n$. The strength $S$ of the signal at the receiver can thus be bounded by 
\[
\begin{array}{ll}
S&\geq P_n \min\left\{1,2s_n^{-\beta}\right\}\\
&= O(1)
\end{array}
\]
Finally, we obtain the following bound on the SINR: $SINR\geq \frac{S}{N_o + I}$. As the above expression does not depend on $n$, the theorem is proven.
\end{proof}

\subsection{Uniform speed-limited (USL) mobility} 
\label{subsec:MobSmooth}
Nodes are mobile and move according to the uniform speed-limited (USL) model, 
a fairly general mobility model defined next.
The USL model essentially embodies two conditions: (i) the node distribution
at every time step is uniform over the network domain, and (ii) the distance
a node can travel over a time step is bounded. 
We restrict ourselves to the case in which the maximum
speed is not dependent on $n$. In practice, of course, such an
assumption is realistic since the maximum speed of the nodes will not
increase when new nodes join the network.
\begin{definition}
A collection of $n$ nodes satisfy
the uniform speed-limited (USL) mobility model if the following
two conditions are satisfied:
\begin{enumerate}
\item At every time $t$, the distribution of nodes over the network
domain is uniform;
\item For every node $u$ and time $t$, the distance traveled in
the next time step is bounded, i.e., $||x^{(t+1)}(u)-x^{(t)}(u)|| < S$.
\end{enumerate}
%$x^{(t)}(u)$ at time $t$ moves to a position
%$x^{(t+1)}(u)=x^{(t)}(u)+s^{(t)}$ at time $t+1$. 
%The vector $s^{(t)}$
%is drawn from a symmetric probability distribution such
%that $0<||s^{(t)}||<S$, where $S$ is a constant (bounded speed), and
%$\pr{s^{(t)}}=\pr{s'^{(t)}}$ if $||s^{(t)}||=||s'^{(t)}||$. Finally,
%we consider that nodes bounce off the borders of the square area
%(billiard model).
\end{definition}

The USL mobility model is quite general. For example, it includes the
following cases: (i) The nodes perform independent random walks with  bounded one-step
displacement. The random walks can be biased, and the displacement distribution does not
need to be homogeneous over the node population. We have to assume that the nodes operate
in the stationary regime. (ii) The nodes follow the random waypoint model (RWP). The system has to be in the stationary
regime. (iii) The generalized random direction models from \cite{Sharma06}, which interpolate
between the random walk and the random waypoint cases, through a control parameter that
can be viewed as the "locality" of the mobility process. (iv) We can also allow for models where nodes do not move independently. As an illustrative
example, assume we uniformly place nodes on the square; the nodes then move in lockstep
according to any speed-limited mobility process, maintaining their relative positions to each
other. Observe that the uniform distribution is maintained for all time steps (note that we
move on a torus), and that the speed-limited property is true by definition.
\par
We see that the USL class of mobility models is fairly general, and includes many of 
the models that have been proposed in the literature. For simplicity, we consider that time is discrete. In other words, we
look at a snapshot of the network every $\Delta T$ seconds. At every
time step, the connectivity between nodes will be modified. Hence, we
will work with a sequence of connectivity graphs. In order to design a
routing algorithm with a low control traffic overhead, we will need to
understand how fast the graph distances between nodes can evolve over
time. In particular, consider two nodes $u$ and $v$ at distance
$d=d^{t}(u,v)$ at time $t$. We want to bound the multiplicative factor
by which this distance can change in $\kappa$ time steps. Formally, we
define $\kappa(\tau,d)$ as follows:
\begin{definition}
\label{def:kappa}
We say that a communication network is $\kappa(\tau,d)$-smooth if the
shortest path distance between any two nodes $u$ an $v$ at shortest
path distance $d$ cannot change by more than a factor $\kappa(\tau,d)$
in $\tau$ time steps \emph{i.e., }$\forall u,v$, we have:
\[\max\left\{\frac{d^{(t)}(u,v)}{d^{(t+\tau)}(u,v)},\frac{d^{(t+\tau)}(u,v)}{d^{(t)}(u,v)}\right\}\leq\kappa(\tau,d)\] 
\end{definition}
Additionally, we simply say that the network is \emph{$\kappa$-smooth}
if there exists a constant $\nu$ such that $\kappa(\nu
d,d)\leq\kappa(\nu)=\kappa$ independently of $d$. In this case, the
distances grow at the same speed at all scales. In the sequel, we will
bound $\kappa$ and $\nu$ for our model. This USL property holds for a general class of {\em random trip}
mobility models studied in \cite{LeBoudec2005}, where it is shown that
the stationary distribution of such mobility models is uniform and
ergodic. We restate this theorem without proof.
\begin{theorem}(\cite{LeBoudec2005}) The random-trip mobility model
has uniform stationary distribution on $[0,a)\times[0,a)$.
\end{theorem}
\subsection{Assumptions}
\label{sub:Ass}
We consider that a time step $\Delta T$ is much larger than the round
trip time (RTT) through the network \emph{i.e., } the time scale for
mobility is much larger than the time scale for communications. For
clarity and in order to simplify the analysis, we will make the
assumption that nodes can communicate instantaneously through the
network. We also make the assumption that there is a random
permutation $\pi$ on the nodes, and that all nodes in the network know
their rank in the permutation. In Section \ref{sec:impl} we will then
drop these assumptions and consider practical aspects of the
implementation. Finally, we say that a result holds with high
probability (w.h.p.) if it holds with probability at least
$(1-O(\frac{1}{n^{\rho}}))$, for some constant $\rho>0$. In Table
\ref{tab:not}, we summarize the notations used in this paper.
\begin{table}[htb]
\centering
\begin{tabular}{|c|c|}
\hline $x^{(t)}(u)$&Position of node $u$ at time $t$\\ $d^{(t)}(u,v)$&
Shortest path distance from $u$ to $v$ at time $t$\\ $r_n$&Wireless
communication radius\\ \G{n}{r_n}&Random geometric graph\\
\B{(t)}{u}{R}&Ball of radius $R$ around $u$\\
$\kappa(\tau,d)$&$\max\left\{\frac{d^{(t)}(u,v)}{d^{(t+\tau)}(u,v)},\frac{d^{(t+\tau)}(u,v)}{d^{(t)}(u,v)}\right\}\leq\kappa(\tau,d)$\\\hline
\end{tabular}
\caption{Table of notations}
\label{tab:not}
\end{table}

\section{Network Properties}
\label{sec:npp}

In this section, we prove some properties of the network models
presented in Section \ref{sec:mod}, which are necessary to analyze
the performance of our algorithm. We focus our attention on the
geometric random graph $\G{n}{r_n}$, but all the arguments can be
extended to the SINR full connectivity model with TDMA scheduling,
discussed in Section \ref{subsec:SINRmodel}. In particular, for
$\G{n}{r_n}$, we now consider the case in which the communication
radius $r_n$ is such that
$r_n=\sqrt{(1+\epsilon)\log{n}}>\log^{1/2}n$, where $\epsilon>0$.

\par For uniform speed-limited (USL) mobility models discussed in
Section \ref{subsec:MobSmooth}, at each time, the node locations
$\{x^{(t)}(u)\}$ have a distribution that is uniform over
$\left[0,\sqrt{n})\right.\times \left[0,\sqrt{n})\right.$. Therefore,
we now discuss the property of a sequence of geometric random graphs,
$\mathcal{G}^{(t)}(n,r_n)$, under USL mobility model. We subdivide the network area on
which the nodes live into smaller squares of side $\frac{r_n}{c}$,
where $c$ is a constant chosen such that nodes in neighboring squares
are connected (see Fig. \ref{fig:squarel})
and that an integer number
of squares fit into the network area.
\begin{figure}[htbp]
\centering
\includegraphics[height=3cm,keepaspectratio]{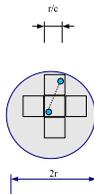}
\caption{Nodes in neighboring squares are connected}
\label{fig:squarel}
\end{figure}   
We arbitrarily set $c=\sqrt{5}$. We number the small squares from $1$
to $m=\frac{n}{(r_n/c)^2}=\frac{nc^2}{(1+\epsilon)\log{n}}$ and denote
by $E_i$ the event that small square $i$ does not contain any node, in a
sequence of length $n^{\rho}$ time steps, for some constant $\rho$. In
the next theorem, we show that when nodes move according to USL
mobility model, all small squares will be populated w.h.p. 
\begin{theorem}
\label{lem:seqcg}
There exists a constant $\rho$ such that if we divide the network into
small square of side $\frac{r_n}{c}$ (with $r_n>\sqrt{\log{n}}$), at every time step in a sequence of length $n^{\rho}$, 
every small square contains at least one node w.h.p. 
\end{theorem}
\begin{proof}
Consider a sequence of length $Z=n^{\rho}$. Denote by $E_i^{(j)}$ the
event the small square $i$ is empty at time $j$. Let
$m=\frac{n}{r_n^2}$. We can compute:
\[
\begin{array}{ll} 
\pr{\bigcup_{j=1}^{Z}\bigcup_{i=1}^{m}E_i^{(j)}}&\leq Z\sum_{i=1}^{m}\pr{E_i^{(j)}}\\
&=Z\sum_{i=1}^{m}(1-\frac{1}{m})^n\\
&\leq Z \sum_{i=1}^{m}e^{-\frac{n}{m}}\\
&=Z\frac{nc^2}{(1+\epsilon)\log{n}}e^{-\frac{nc^2(1+\epsilon)\log{n}}{n}}\\
&\leq Z\frac{nc^2}{(1+\epsilon)\log{n}}\frac{1}{n^{(1+\epsilon)}}\\
&=Z\frac{c^2}{(1+\epsilon)n^{\epsilon}\log{n}}\\
&\leq O(\frac{n^{\rho}}{n^{\epsilon}})\\
&=O(\frac{1}{n^{\epsilon-\rho}})
\end{array}
\]
We can now choose $\rho$ such that $\epsilon-\rho>0$ and the result follows.
\end{proof}
It is immediate that a in single instantiation of the connectivity graph (\emph{i.e., }time step),
every small square is populated w.h.p.
\begin{corollary}
\label{lemma:fs}
With probability at least $(1-O(\frac{1}{n^{\epsilon}}))$, there is no empty small square in a sequence of length $1$.
\end{corollary} 
\iffalse
\begin{proof}
Follows directly from Theorem \ref{lem:seqcg} when we take a sequence
of length $n^{\rho}=1$.
\end{proof}
\fi

We are now ready to show that at every time step in a sequence of $n^{\rho}$ connectivity graphs, the connectivity graph is doubling w.h.p.
Since we have a USL mobility model, any graph
$\mathcal{G}^{(t)}(n,r_n)$ is statistically identical to
$\G{n}{r_n}$.

\begin{theorem}
\label{thm:rlar}
\G{n}{\sqrt{(1+\epsilon)\log{n}}} are doubling w.h.p.
\end{theorem}
\begin{proof}
By Lemma \ref{lemma:fs}, all small squares contain at least one node
w.h.p. Consequently, neighboring squares (vertically and horizontally)
have at least one communication link. Denote by \g{m}{r} the grid
having the small squares as vertices, and with edges between vertical
and horizontal neighbors. Consider a ball
$B_{upper}=\B{\mathcal{G}}{u}{2R}$ centered around some node
$u$. Clearly, all nodes in $B_{upper}$ must be contained in a square
which is part of $\B{\mathcal{L}}{u}{4Rc}$ \emph{i.e.,}
$B_{upper}\subseteq \B{\mathcal{L}}{u}{4Rc}$. This follows from the
fact that no node in $B_{upper}$ can be further away from $u$ than
$2Rr$ in Euclidean distance, and that the grid is fully connected
w.h.p. Similarly, one can see that $\B{\mathcal{L}}{u}{R}\subseteq
B_{lower}=\B{\mathcal{G}}{u}{R}$. This is a consequence of the fact
that $\mathcal{L}$ is a subgraph of $\mathcal{G}$, \emph{i.e.,} two
nodes in small squares $R$ hops a part in $\mathcal{L}$ cannot be more
than $R$ hops apart in $\mathcal{G}$ (see Fig. \ref{fig:gridCov}). For an appropriately chosen constant $\alpha$, we have:
\begin{equation}
B_{upper}\subseteq \B{\mathcal{L}}{u}{4Rc} \subseteq
\bigcup_{j=1}^{\alpha} \B{\mathcal{L}}{v_j}{R} \subseteq
\bigcup_{j=1}^{\alpha} \B{\mathcal{G}}{v_j}{R}
\end{equation}
and \G{n}{\sqrt{(1+\epsilon)\log{n}}} is \textit{doubling}.
\begin{figure}[htbp]
\centering
\includegraphics[width=0.7\textwidth,keepaspectratio]{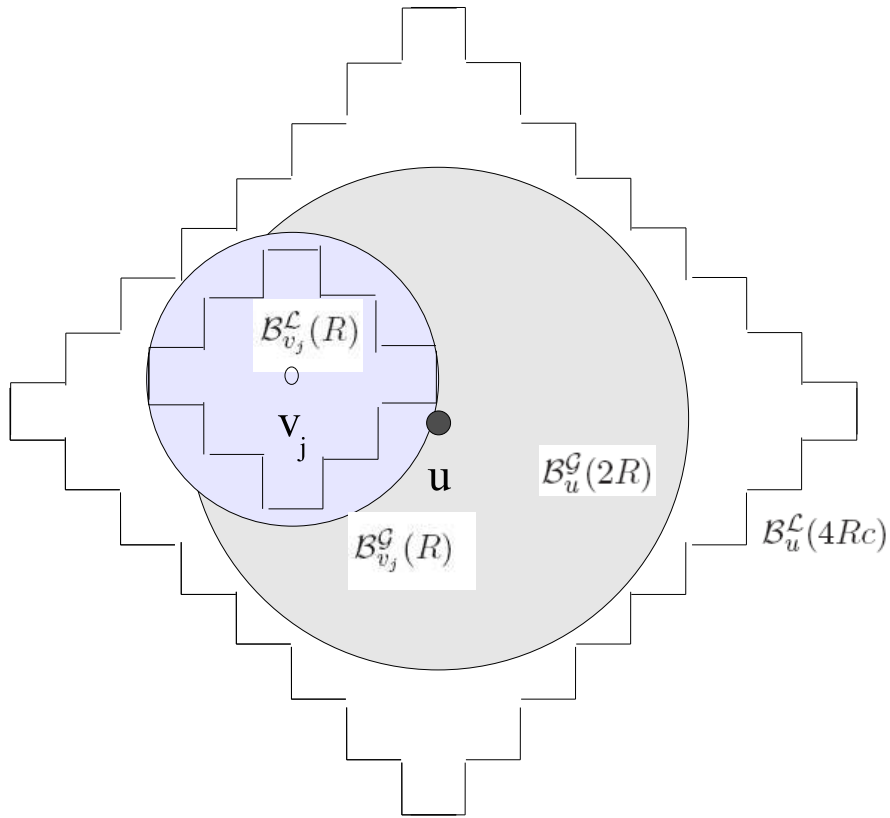}
\caption{Proof of theorem \ref{thm:rlar}}
\label{fig:gridCov}
\end{figure}   
\end{proof}
Note that it is possible to build a deterministic geometric graph for
which this property does not hold (see Appendix \ref{apx:A}). Further,
one can show that \G{n}{r_n} are \emph{not} doubling w.h.p when
$r_n<\sqrt{\log n}$. We prove this result in Appendix \ref{apx:B}. At
this point, we would like to emphasize that even though we
analyze networks in which the nodes are uniformly distributed on a
square area, the doubling property is a much more powerful
tool. Indeed, our results and algorithms depend only on the doubling
constant. Consequently, the algorithms and the bounds can be applied
to any other type of networks or node configuration which lead to a
doubling connectivity graph. For instance, one can consider the
network shown in Figure \ref{fig:wall}, described in Section
\ref{subsec:GnR}. It can easily be shown by using a technique similar
to the one used in Theorem \ref{thm:rlar} that this network is
doubling. While we can seamlessly apply our routing algorithm to such
a network, any classical geographic routing algorithm would fail or
require a high control traffic overhead to get out of dead-ends. This is because nodes would get stuck against the wall when
routing packets from the lower to the upper part of the network. In
turn, this would considerably degrade the performance in terms of
stretch and control traffic overhead with respect to the same network
without a wall. In the next subsection we prove a set of sufficient conditions for a wireless networks to have a constant doubling dimension.

\subsection{Inhomogeneous Topologies}
In the first part of this subsection, we show that under certain conditions, the presence of topological holes (obstacles) in the network does not increase the doubling property, or only by a constant factor. In particular, we are interested in how we can alter the topology of a fully connected and dense network by removing nodes while still preserving the doubling property. In the second part, we will generalize this idea to arbitrary metric spaces. Consider a \G{n}{r_n} with $r_n>\sqrt{\log n}$, such that full connectivity is guaranteed. The network area is divided into squarelets of side $\frac{r_n}{c}$, where $c$ is chosen such that nodes in horizontally and vertically adjacent squarelets are guaranteed to be within communication range. We now arbitrarily select squarelets and remove all nodes they contain. We denote the new graph we obtain by \Gn.  
\begin{figure}[htbp]
	\centering
	\includegraphics[width=0.8\textwidth,keepaspectratio]{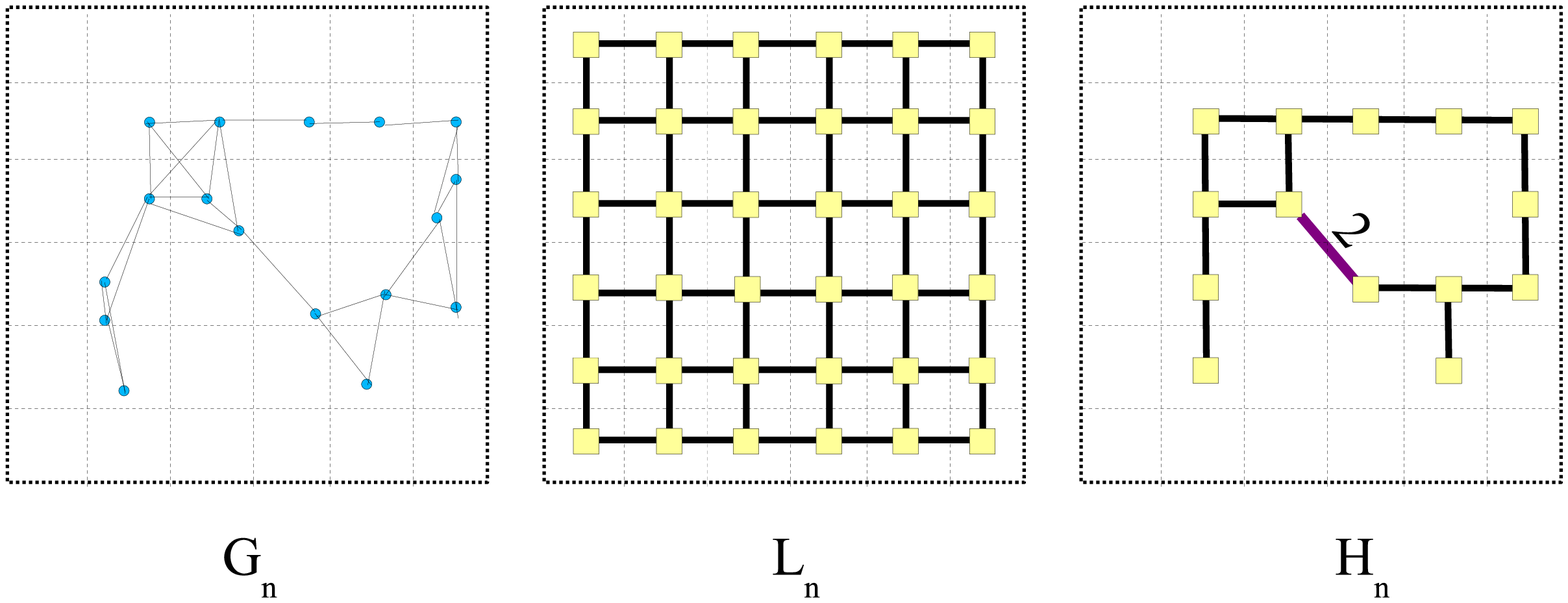}
	\caption{Graphs \Gn, \Ln and \Hn. The network area is divided into squarelets of side $\frac{r_n}{c}$ such that nodes in horizontally and vertically adjacent squarelets are guaranteed to be within communication range.}
	\label{fig:GridHole}
\end{figure} 
We denote by \Ln the full grid where the squareletes are vertices and by \Hn the corresponding grid in \Gn \emph{i.e.,} the thinned out grid obtained by selecting only non-empty squarelets in \Gn. In \Ln, we add an edge between horizontally and vertically adjacent squarelets (see Fig. \ref{fig:GridHole}). In \Hn, we first add a an edge between horizontally and vertically squarelets containing at least one node. Then, for every pair of squarelets containing nodes that can communicate directly, we add an edge of weight corresponding to the distance between those two squarelets in \Ln. We add the edges from the shortest to the longest one, and only if no path of the same length already exists in \Hn. 
We can now define a topological hole as follows: 
\begin{definition}[Topological Hole]
A set of horizontally, vertically and horizontally adjacent empty squarelets in the graph \Hn is called a \emph{hole} if adding a (virtual) vertex in all of the squarlets in that set modifies the distance between at least two vertices in \Hn.  
\end{definition}
Let us denote by \void{k} the $k^{th}$ hole (k=1,2,3,...). We define the perimeter $p(\void{k})$ of \void{k} as 2 times the maximum distance between any two vertices on the border of the hole \emph{i.e,} in squarelets adjacent to the empty squarelets defining the hole. Note that for all $u,v$, we have $d^{\Hn}(s_u,s_v)\geq d^{\Gn}(u,v)$.
\begin{figure}[htbp]
\centering
\includegraphics[width=0.8\columnwidth,keepaspectratio]{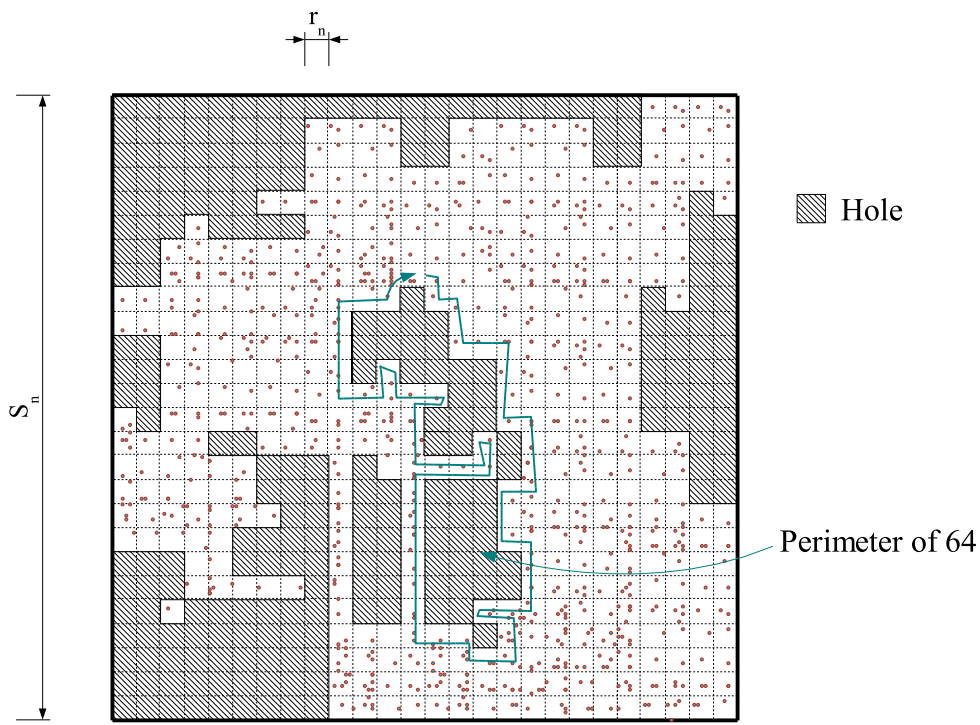}
\caption{A network with holes}
\label{fig:network}
\end{figure}  
\begin{theorem}
\label{thm:holes}
Let $P=\max_k p(\void{k})$. Then, the doubling dimension is upper bounded by $O(P^{2})$.
\end{theorem}
\begin{proof}
Consider a ball $\B{\Gn}{u}{2R}$ centered at $u$ in \Gn. First, observe that $\B{\Gn}{u}{2R}\subseteq \Boxes{\Ln}{u}{2Rc}$, where $\Boxes{\Ln}{u}{2Rc}$ is the box centered at the squarelet containing $u$ in \Ln which contains all nodes at ``maximum norm'' $2Rc$ (\emph{i.e.,} $l_\infty$-norm) from this squarelet. In other words, all nodes within $2R$ hops from $u$ in \Gn must be in a squarelet contained in this box. We will now cover this box with smaller boxes $\Boxes{\Ln}{s_{v_i}}{\max\left\{1,\left\lceil \frac{R}{4\gamma}\right\}\right\rceil}$. We need $\left\lceil \frac{16R^2c^24\gamma^2}{R^2}\right\rceil=64c^2\gamma^2$ such boxes at most. Consider the same small boxes in \Hn. Pick one non-empty squarelet $s_{v_i}$ in each such small boxes. Note that the maximum hop distances between two squarelets in such a small box in \Ln is at most $\frac{R}{\gamma}$. For each of these hops, we might have to make a detour of at most $P$ steps. Consequently, the same two squarelets could be at distance at most $\frac{PR}{\gamma}$ in \Hn. Observe now that for any two nodes $u$ and $w$ contained in squarelets $s_u$ and $s_w$ respectively, we have $d^{\Hn}(s_u,s_w)\geq d^{\Gn}(u,w)$.  For each squarelet $s_{v_i}$, we pick one node $v_i$ contained in this squarelet. Hence, for all nodes $w$ contained in this small box, we have $d^{\Gn}(v_i,w)\leq d^{\Hn}(s_{v_i},s_w)\leq \frac{PR}{\gamma}$. By setting $\gamma=P$, we obtain the claim.   
\end{proof}
We can extend this result to the case where the network can be divided into convex sets. We define a convex set in \Gn with slack as follows:
\begin{definition}
Let $\Psi$ be a set of nodes in \Gn. Let $\Hn^{(\Psi)}$ be the squarelets in \Hn containing at least one node in $\Psi$. We say that the set $\Psi$ is \emph{convex} if $\forall u,v\in \Psi$, $d^{\Hn^{(\Psi)}}(s_u,s_v)=d^{\Ln}(s_u,s_v)$, where $s_u$ and $s_v$ are the squarelets containing $u$ and $v$\emph{ i.e.,} there must be at least one shortest path inside the convex set. We say that the set $\Psi$ is \emph{convex with slack P} if $d^{\Hn^{(\Psi)}}(s_u,s_v)\leq Pd^{\Ln}(s_u,s_v)$.%it were convex if we ``filled'' all topological holes of perimeter smaller or equal to $P$.
\end{definition}
We can now state the following theorem
\begin{theorem}
\label{thm:holes2}
Let $\zeta_1, \zeta_2,...,\zeta_q$ be a partition of the network into $q$ convex sets with slack $P_1,P_2,...,P_q$ respectively. Let $per_1,per_2,...,per_q$ denote the perimeter of the convex sets. The doubling dimension is then upper bounded by $\max_{u,\rho}4\sum_{\zeta_i:\zeta_i\cap\Boxes{\Ln}{s_u}{\rho}\neq \emptyset}(\left\lceil \frac{per_i}{\frac{\rho}{P_i}}\right\rceil)^2$.  
\end{theorem}
\begin{proof}
In the proof of Theorem \ref{thm:holes}, we have shown that any ball of radius $2R$ around some node $u$ is contained in a box \Boxes{\Ln}{s_u}{\rho}, where $\rho=2Rc$. We can cover each convex set $\zeta_i$ intersecting this box with at most $4(\left\lceil \frac{per_i}{\frac{\rho}{P_i}}\right\rceil)^2$ small boxes of radius $R/4P_i$, as shown in Theorem \ref{thm:holes}\footnote{A convex area of perimeter $q$ can always be included in a square area of side $q$.}. A slack of $P$ implies that by selecting one node in each of the small boxes, all nodes in the convex set are within $R$ hops of this node in \Gn. If the box \Boxes{\Ln}{s_u}{\rho} is partitioned into several convex sets, selecting $4(\left\lceil \frac{per_i}{\frac{\rho}{P_i}}\right\rceil)^2$ nodes in each convex set $\zeta_i$ intersecting this box will in turn guarantee that all nodes are covered. 
\end{proof}
In practice, this result implies that if we are given a decomposition of the network into convex sets, we can bound the overall doubling dimension given the doubling dimension of each set separately. Further, this result implies that networks that consist of a small number of convex areas, which can each contain arbitrarily many small holes, have a low complexity in terms of doubling dimension. We will now relate the ``shapes'' of a topological hole to the doubling dimension. In particular, we will show that one can relate the doubling dimension to the maximum number of connected components in any square subarea. 
\begin{theorem}
\label{thm:holes3}
For any $\gamma\geq 2$, the doubling dimension $\alpha$ is such that 
\[
\alpha\leq 4\gamma^2c^2 \max_{\Boxes{\Ln}{u}{R/\gamma}}\left\{\mbox{number of convex disconnected components with slack $\gamma$ in \Boxes{\Ln}{u}{R/\gamma}}\right\}
\]
\end{theorem}
\begin{proof}
In the proof of Theorem \ref{thm:holes}, we have shown that any ball of radius $2R$ around some node $u$ is contained in a box \Boxes{\Ln}{s_u}{\rho}, where $\rho=2Rc$. In turn, we showed that by dividing this box into smaller boxes of side $R/\gamma$, and by selecting one node in each box, we could cover the larger ball of radius $2R$. Now, in each small box of side $R/\gamma$, the presence of holes might create several disconnected components. However, we know that inside each such component, we can cover any convex subset with slack $\gamma$ with one nodes. The result follows. 
\end{proof}
This last result gives us a characterization of the alterations we can make to a fully connected \G{n}{r_n} network, while only affecting the doubling dimension by a constant factor. In particular, we can remove nodes as long as we do not create too many convex and disconnected components in any square subarea. Note that we can still remove arbitrarily many nodes as long as we only create small holes. Theorems \ref{thm:holes}, \ref{thm:holes}, \ref{thm:holes} imply that topologies such as the one shown in Fig. \ref{fig:DoublingNet} have a constant doubling dimension. 
\begin{figure}[htbp]
	\centering
	\includegraphics[width=0.8\textwidth,keepaspectratio]{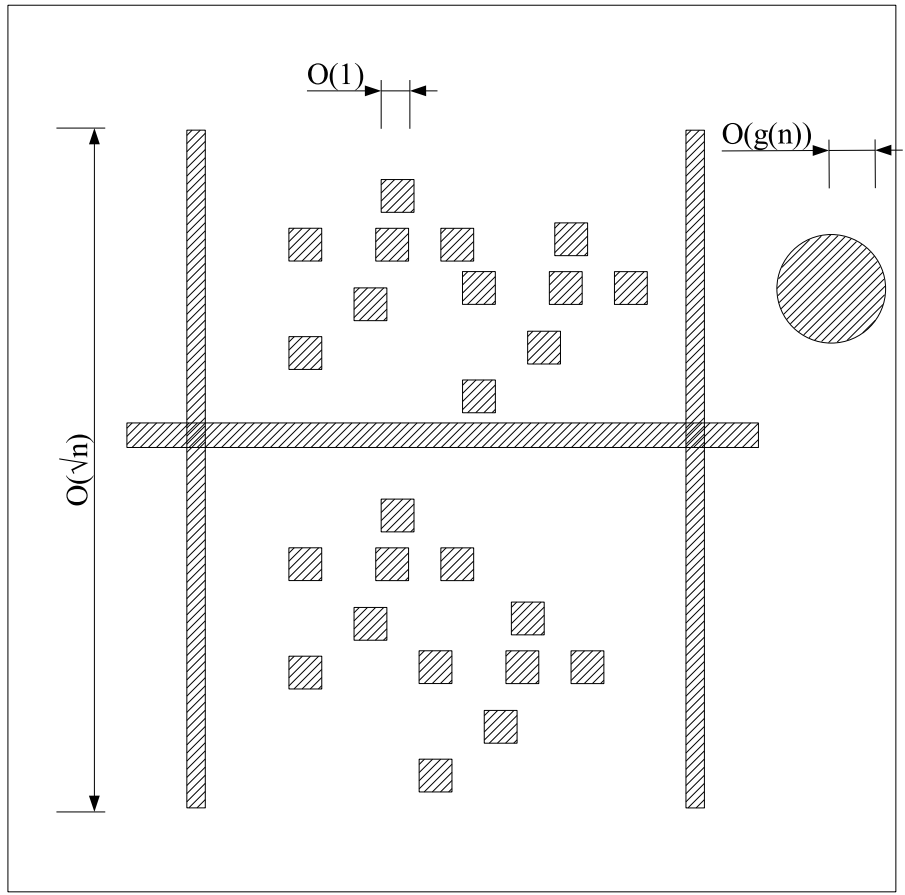}
	\caption{A network with topological holes and a constant doubling dimension. The size of the large holes grow with $n$, but the network can be divided into a constant number of areas, each being convex with slack $O(1)$ \emph{i.e.,} each of the convex areas contains only obstacles with a constant perimeter or that can only increase the distance between nodes by a constant factor. Note that even though the doubling dimension is low, greedy geographic forwarding of packets would fail as packets would get stuck in dead-ends against the holes. Squarelets containing no nodes are hatched.}
	\label{fig:DoublingNet}
\end{figure} 
The results stated above are special cases of the more general result detailed in the sequel. Indeed, we can relate the doubling dimension in a metric space to the doubling dimension in another metric space if we know the distortion of the embedding that maps the points in one metric space to the points in the other metric space. The example above is a special case of that setup where we map the nodes of a graph to points in Euclidean space. Consider two metric spaces $(X,d)$ and $(X',d')$, where $d$ and $d'$ are distance functions which define a metric on the sets of point $X$ and $X'$. We could for instance consider the two metric spaces $(\mathcal{X},||.||)$ and $(\mathcal{H},d(.,.))$ \emph{i.e.,} the points in the plane with the Euclidean distance and the nodes in the graph with the shortest path distance. A \emph{metric embedding} is a bijective function $\phi:\mathcal{X}\rightarrow\mathcal{X'}$ which associates to a point in one metric space a point in another metric space. 
\begin{definition}[Distortion of an Embedding]
\label{def:doe}
A mapping $\phi:X\rightarrow X'$ where $(X,d)$ and $(X,d')$ are metric spaces, is said to have distortion at most $D$, or to be a \emph{D-embedding}, where $D\geq 1$, if there is a $K\in (0,\infty)$ such that $\forall x,y \in X$,
\[
Kd(x,y)\leq d'(\phi(x),\phi(y)) \leq KDd(x,y)
\]
if $X'$ is a normed space, we typically require $K=1$ or $K=\frac{1}{D}$. An embedding has distortion $D$ with slack $\epsilon$ if all but an $\epsilon$ fraction of node pairs have distortion $D$ under $\phi$. Additionally, one can loosen this definition by allowing slack. The slack is said to be \emph{uniform} if each node has distortion at most $D$ to a $1-\epsilon$ fraction of the other nodes. Finally, an embedding with distortion $D$ and slack $\epsilon$ is \emph{coarse} if for every node $u$ the distortion is bounded to a node a distance greater than $r_\epsilon=inf\left\{ r \mbox{s.t } |\B{X}{u}{r}|>\epsilon n\right\}$.  
\end{definition}  
The doubling dimension of a metric space embedded into another metric space can be bounded as follows:
\begin{theorem}[Bounding the Doubling Dimension]
\label{thm:embed}
Consider a metric space $(\mathcal{H},d)$ embedded in another metric space $(\mathcal{E},d')$ by a function $\phi$. Let the doubling dimension of $\mathcal{E}$ be $\beta$. Let the distortion of this embedding be $D$. Then, $\mathcal{H}$ has doubling dimension  $\alpha$ with %$O((\frac{1}{D})^{\log\beta})\leq
$\alpha\leq O((2D)^{\log\beta})$.
\end{theorem}
\begin{proof}
\begin{figure}[htbp]
	\centering
	\includegraphics[width=0.8\textwidth,keepaspectratio]{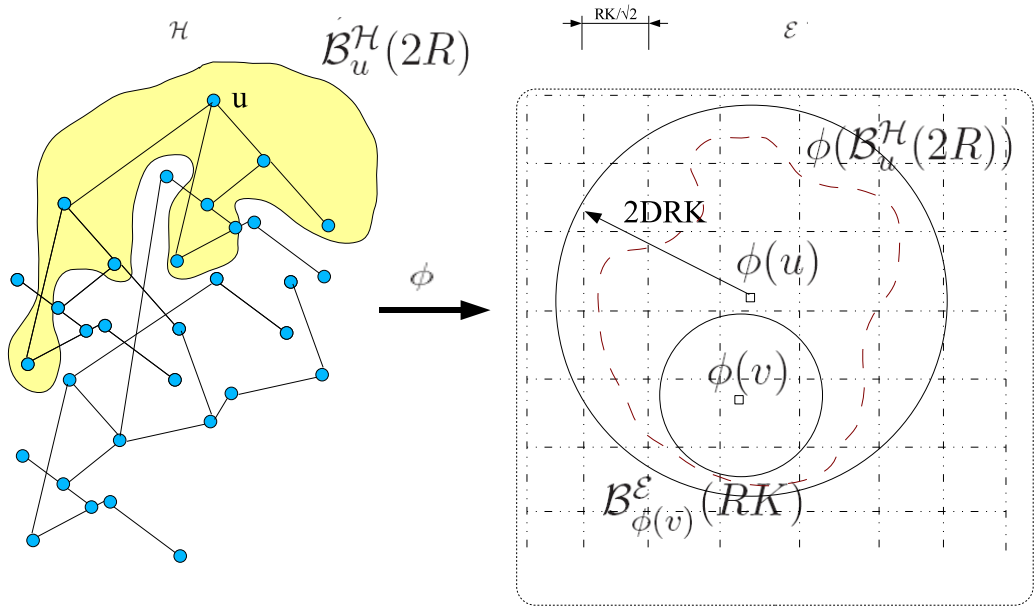}
	\caption{Proof of Theorem \ref{thm:embed}}
	\label{fig:embed}
\end{figure}   
Choose any node $u\in\mathcal{H}$. If the above condition is fulfilled, the images of all nodes in $\B{\mathcal{H}}{u}{2R}$ can be at distance $d'$ at most $2KDR$ from $u$ at $\phi(u)$. Hence, $\phi(\B{\mathcal{H}}{u}{2R})\subset \B{\mathcal{E}}{\phi(u)}{2KDR}$. We will now try to cover $\phi(\B{\mathcal{H}}{u}{2R})$ by as few balls $\B{\mathcal{\mathcal{E}}}{\phi(v)}{RK}$ as possible (see Fig. \ref{fig:embed}, which illustrates this setup in the case when $\mathcal{H}$ is a graph and $\mathcal{E}$ the Euclidean space). To do so, let us cover $\B{\mathcal{E}}{\phi(u)}{2KDR}$ by small balls of radius $KR$ in $\mathcal{E}$. Covering $\B{\mathcal{E}}{\phi(u)}{2KDR}$ will require at most $\beta^{\log 2D}$ balls of radius $KR$ in $\mathcal{E}$, given that $\mathcal{E}$ has doubling dimension $\beta$. We know that $d(u,v)\leq d'(u,v)/K$, by definition \ref{def:doe}. Consequently, $\phi^{-1}(\B{\mathcal{\mathcal{E}}}{\phi(v)}{RK})\subset \B{\mathcal{H}}{v}{R}$. We can conclude that $\B{\mathcal{H}}{u}{2R}\subset \bigcup_{j=1}^{\beta^{\log 2D}}\B{\mathcal{H}}{v_j}{R}$. 
%\par
%The lower bound can be proven in a very similar way. Indeed, one can observe that $\phi^{-1}(\B{\mathcal{\mathcal{E}}}{\phi(u)}{2RK})\subset \B{\mathcal{H}}{u}{2R}$ and $\phi(\B{\mathcal{H}}{v}{R})\subset \B{\mathcal{\mathcal{E}}}{\phi(v)}{RKD}$. By using a similar argument as we did for the upper bound, one can show that $O(\frac{1}{D^2})$ balls $\B{\mathcal{\mathcal{E}}}{\phi(v)}{RKD}$ at least are necessary to cover $\B{\mathcal{\mathcal{E}}}{\phi(u)}{2RK}$.
\end{proof}
The presence of large obstacles in the network does not necessarily imply that the network is not doubling. In particular,
\begin{theorem}
Consider a metric space $\mathcal{E}$ with doubling dimension $\beta$. A metric space $\mathcal{H}$ that can be divided in $k$ sets $S_1,S_2,...,S_k$, such that each set embeds individually with distortion $D_i$ into $\mathcal{E}$ has doubling dimension at most $\sum_{j=1}^{k}\beta^{2\log2D_j}$.
\end{theorem}
\begin{proof}
Consider any ball of radius $2R$ in $\mathcal{H}$, such that the nodes in the ball belong to at least two different sets (otherwise the theorem is clearly true). Note that the radius of each of these subsets can be at most $4R$. Consequently, we now that the part of the ball that belongs to $S_i$ can be covered by at most $\beta^{2\log2D_i}$ (by applying Theorem \ref{thm:embed} to cover a ball of radius $4R$ by balls of radius $R$). The theorem follows.
\end{proof}
We can now broaden the class of communication networks that have low doubling dimension. In particular, if we can subdivide the communication graph into a constant number of subsets, such that each one embeds with constant distortion into the Euclidean plane, the whole network is doubling. Consequently, topologies such as the one shown in Figure \ref{fig:DoublingTop} are doubling. In this example, we embed an unweighted graph into the Euclidean plane. Note that the minimal Euclidean distance between nodes should be $\rho r_n$ (for some constant $\rho$), such that $\rho r_n d(u,v)\leq ||x(u)-x(v)||\leq O(1) \rho r_n$. If this equation is true for all pairs of nodes, then the distortion is $O(1)$. There is an issue when the nodes are neighbors in the communication graph, as the above rule implies that the Euclidean distance between such pairs of nodes should then be at least $O(r_n)$. However, we can ignore the distances below $2$ as we will not cover balls of radius 1 (since we have a broadcast medium, the degree of a node does not impact the communication overhead). In such cases, it is obvious that geographic routing would fail, even though the inherent complexity of the network is low. Indeed, packets would get stuck against walls. Remarkably, our routing algorithm is oblivious to the topology and only depends on the doubling dimension. Hence, there is absolutely no need to detect or identify obstacles. The communication overhead will simply depend on the doubling dimension.
\begin{figure}[htbp]
	\centering
	\includegraphics[width=0.8\textwidth,keepaspectratio]{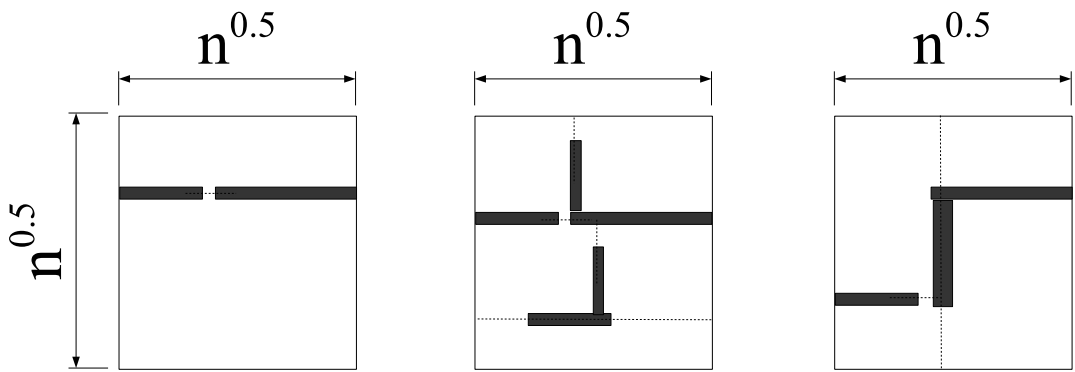}
	\caption{A set of doubling network topologies. The network is dense, and made inhomogeneous by the walls, which do no allow transmissions to go through. Note that the walls stretch when $n$ grows, such that the network wide distortion also grows with $n$. Dashed lines indicate the separation into sets.}
	\label{fig:DoublingTop}
\end{figure}

\subsection{Sequences of Communication Graphs}
In this subsection we study the behavior of a sequence of communication graphs, \emph{without} any obstacles. We show that a sequence of $\mathcal{G}^{(t)}(n,r_n)$ of
length $n^{\rho}$, for some constant $\rho$, with the USL mobility
model is $\kappa$-smooth. As already seen in Theorem \ref{thm:rlar},
such a sequence of graphs is doubling at every time instant.
\begin{theorem}
\label{cor:rwk}
A sequence of $\mathcal{G}^{(t)}(n,r_n)$ of length $\leq n^{\rho}$,
where nodes move according to the USL mobility model with maximum
constant speed $S$ is
\[
max\left\{\frac{r_nd^{(t)}}{\frac{r_nd^{(t)}}{\sqrt{5}\sqrt{2}}-2\sqrt{5}\sqrt{2}\tau
S},\sqrt{5}\sqrt{2}(1+\frac{2\tau
S\sqrt{5}\sqrt{2}}{r_nd^{(t)}})\right\}
\]
smooth w.h.p.
\end{theorem}
\begin{proof}
Consider two nodes $u$ and $v$ at Euclidean distance
$q_2^{(t)}=||x_{u}-x_{v}||_2$ at time $t$. Let
$q_1^{(t)}=||x_{u}-x_{v}||_1=\sum_{m=1}^{2}|x_m(u)-x_m(v)|$. Further,
denote by $d^{(t)}=\oft{d}(u,v)$ their shortest path distance at time
$t$. One can see that $\frac{q_2^{(t)}}{r_n}\leq \oft{d} \leq
\frac{\sqrt{5}\sqrt{2}q_2^{(t)}}{r_n}$.  Indeed, the shortest possible
path will follow a straight line between $u$ and $v$. The length of
this line is $q_2^{(t)}$ and one hop can be of length at most
$r_n$. In the worst case, the shortest path from $u$ to $v$ will
follow the shortest path in the grid formed by the small squares of
side $\frac{r_n}{c}=\frac{r_n}{\sqrt{5}}$, which exists w.h.p. Recall
that we can only guarantee horizontal and vertical connectivity
between small squares. The number of small squares in this path will
be at most $\frac{\sqrt{5} q_1^{(t)}}{r_n}$. One can easily show that
$q_1^{(t)}\leq \sqrt{2} q_2^{(t)}$. Let
$x=\left(\begin{array}{c}x_1\\x_2\end{array}\right)
=\left(\begin{array}{c}x_1\\sx_1\end{array}\right)=(x_u-x_v)$. We have
\[
\begin{array}{ll}
q_2^{(t)}&=\sqrt{x_1^2+x_2^2}
=x_1\sqrt{1+s^2}\\
&=(1+s)x_1\frac{\sqrt{1+s^2}}{1+s}
=q_1^{(t)}\frac{\sqrt{1+s^2}}{1+s}.
\end{array}
\] 
Since, we have $\frac{q_2^{(t)}}{q_1^{(t)}}=\frac{\sqrt{1+s^2}}{1+s}$,
the term is maximized when $s=1$. In Figure \ref{fig:sqrt2}, we
illustrate this point.
\begin{figure}[htbp]
\centering
\includegraphics[width=0.9\textwidth,keepaspectratio]{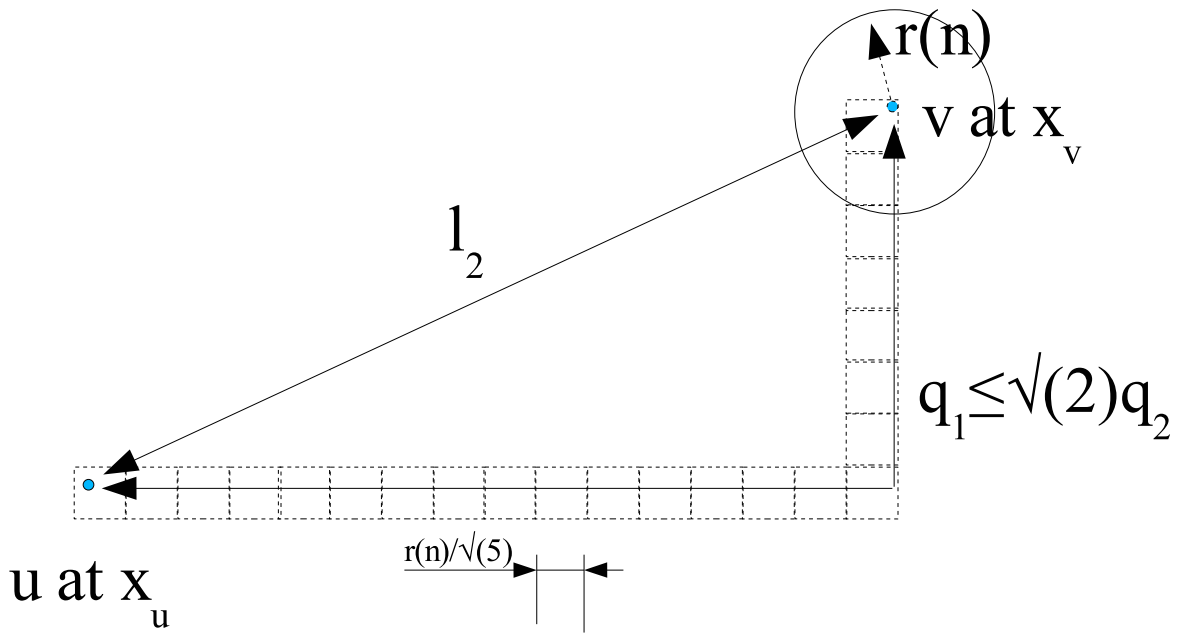}
\caption{Upper and lower bounds for the shortest path}
\label{fig:sqrt2}
\end{figure}  
Similarly, at time $t+\tau$, the shortest path distance will be
bounded by $\frac{q^{(t+\tau)}_2}{r_n}\leq
d^{(t+\tau)} \leq
\frac{\sqrt{5}\sqrt{2}q^{(t+\tau)}_2}{r_n}$. However, we know that the
Euclidean distance can change by at most $2\tau S$ in $\tau$ time
steps\footnote{One can show that this remains true even if the nodes
are reflected on the borders of the network}. Consequently,
\begin{equation}
\frac{q_2^{(t)}-2\tau S}{r_n}\leq d^{(t+\tau)} \leq
\frac{\sqrt{5}\sqrt{2}(q_2^{(t)}+2\tau S)}{r_n}
\end{equation}
We can now bound the multiplicative stretch as follows: 
\iffalse
\begin{equation}
\begin{array}{lll}
 \frac{q_2^{(t)}-2\tau S}{\sqrt{5}\sqrt{2}q_2^{(t)}}&\leq
\frac{d^{(t+\tau)}}{\oft{d}}&\leq \frac{\sqrt{5}\sqrt{2}(q_2^{(t)}+2\tau
S)}}{q_2^{(t)}\\ \rightarrow\ \frac{1}{\sqrt{5}\sqrt{2}}-\frac{2\tau
S}{\sqrt{5}\sqrt{2}q_2^{(t)}}&\leq \frac{d^{(t+\tau)}}{\oft{d}}& \leq
\sqrt{5}\sqrt{2}(1+\frac{2\tau S)}{q_2^{(t)}})\\
\end{array}
\end{equation}
\fi
Hence, 
\[
\begin{array}{l}
max\left\{(\frac{1}{\sqrt{5}\sqrt{2}}-\frac{2\tau
S}{\sqrt{5}\sqrt{2}q_2^{(t)}})^{-1},\sqrt{5}\sqrt{2}(1+\frac{2\tau
S)}{q_2^{(t)}})\right\}\\ = max\left\{\sqrt{5}\sqrt{2}\frac{q_2^{(t)}}{q_2^{(t)}-2\tau
S},\sqrt{5}\sqrt{2}(1+\frac{2\tau S}{q_2^{(t)}})\right\}\\
=max\left\{\frac{r_n\oft{d}}{\frac{r_n\oft{d}}{\sqrt{5}\sqrt{2}}-2\sqrt{5}\sqrt{2}\tau
S},\sqrt{5}\sqrt{2}(1+\frac{2\tau
S\sqrt{5}\sqrt{2}}{r_n\oft{d}})\right\} =\kappa(\tau,d)
\end{array}
\]
\end{proof}
One can now observe that the time it takes to multiply the shortest
path distance between two nodes at distance $d$ is proportional to
$d$. Note that the larger the communication radius $r_n$, the smaller
$\kappa$.  Hence, the distance grows at most linearly with time. In
particular, we have:
\begin{corollary}
\label{cor:kappa}
There exist constants $\nu$ and $\kappa$ defined in the proof such
that a sequence of $n^{\rho}$ connectivity graphs, under the USL
mobility model with maximum constant speed $S$, is
$\kappa$-smooth w.h.p.
\end{corollary}
\begin{proof}
By theorem \ref{cor:rwk}, we know that the sequence is
\[max\left\{\frac{r_n
\oft{d}}{\frac{r_n\oft{d}}{\sqrt{5}\sqrt{2}}-2\sqrt{5}\sqrt{2}\tau
S},\sqrt{5}\sqrt{2}(1+\frac{2\tau
S\sqrt{5}\sqrt{2}}{r_n\oft{d}})\right\}\] -smooth w.h.p. Note
that both terms decrease as a function of the communication radius
$r_n$. Hence, we can set $r_n=1$ without decreasing
$\kappa(\tau,d)$. Similarly, both terms go down when the distance
$\oft{d}$ goes up. We can therefore also set $\oft{d}=1$, which is the
smallest possible distance in an unweighted graph. Consequently, if we
set $\tau=\nu \oft{d}=\nu$, we can now write
\[
\kappa(\tau,d)\leq
max\left\{\frac{1}{\frac{1}{\sqrt{5}\sqrt{2}}-2\sqrt{5}\sqrt{2}\nu
S},\sqrt{5}\sqrt{2}(1+2\nu S\sqrt{5}\sqrt{2})\right\}
\]
which is constant for $\nu$ constant.
\end{proof}

\section{Routing Algorithm}
\label{sec:RoutingAlgo}

We develop the routing algorithm and its performance analysis for a
general class of dynamic networks which produce a sequence of doubling
and smooth connectivity graphs. We have seen in Sections \ref{sec:mod}
and \ref{sec:npp} that this applies to a class of wireless
connectivity models with USL mobility. For notational convenience we illustrate the ideas for a sequence $\mathcal{G}^{(t)}(n,r_n)$
geometric random graphs with USL mobility.

We decompose a time step into two phases: a \emph{beaconing} phase and
a \emph{forwarding} phase. In the former phase, a set of routes are
established by letting all or a subset of nodes flood the network at
geometrically decreasing radii and nodes register with beacon
nodes. In the latter phase, this subset of routes is then utilized by
source nodes to efficiently search for the destination. Every node is
equipped with a routing table as shown in Table \ref{tab:RT}.
\begin{table}[htb]
\centering
\begin{tabular}{|c|c|c|c|}
\hline \textit{Node identifier} & \textit{distance}
$\left[hops\right]$ & \textit{level} & \textit{next hop}\\\hline
\vdots & \vdots & \vdots & \vdots\\\hline
\end{tabular}
\caption{Routing Table RT}
\label{tab:RT}
\end{table}
We first describe two procedures used in the beaconing and the routing phase.

\textbf{flood(R,level) procedure:}
When a node $u$ initiates the $flood(R,level)$ procedure, it
broadcasts a \emph{flood packet} as shown in Table \ref{tab:FP} to its
direct neighbors in \G{n}{r_n}. The \textit{hop count }field is
initialized to $0$ and the content of the \textit{Level} field is set to the level of the beacon. How the level of the beacon is determined will be
specified in the sequel. All nodes can compute the maximum hop count
given the value of the level field in the packet.
\begin{table}[htb]
\centering
\begin{tabular}{|c|c|c|c|}
\hline
\textit{Pkt. Type} & \textit{Node Id.} & \textit{Hop Count} & \textit{Level}  \\
$O(1)$ bits & $O(\log n)$ bits & $O(\log n)$ bits & $O(\log \Delta)$ bits \\
\hline
\end{tabular}
\caption{Flood Packet}
\label{tab:FP}
\end{table}
The neighbors which receive this packet, after increasing the hop
count by $1$, add an entry to their routing table for node $u$ if no
entry for the same node and level with lower or equal hop count is present
in the RT. The \textit{next hop} field is set to
the identifier of the node from which the packet was received. The
level field in the routing table is set to the level given in the
packet. In turn, the nodes which got the packet from $u$ broadcast this packet to their neighbors. The latter follow the same procedure and update the routing table if
necessary. The packet is discarded when the hop count reaches the
maximum hop count (which is a function of the level). Note that with
this procedure, every node forwards the packet at most once and the
distance added to the routing table is the shortest path distance in
\G{n}{r_n}. This procedure also allows us to establish a reverse path from all nodes that get the packet back to $u$. Indeed, it suffices for all these nodes to store the identifier of the node from which they received the packet\footnote{The packet for which they modified their routing table}. Further, for any node $v$, that reverse path a shortest path to $u$.

\textbf{probe(relay,destination) procedure:}
This procedure consists in sending a \textit{probe packet} (see Table \ref{tab:PP}) to a ``relay''
node for which the source has an entry in its routing table. The relay
node will set the success bit to $1$ if it has an entry for the
destination and $0$ otherwise. We will make sure that all nodes on the
path between the source and the relay node have an entry for the relay
node in their routing table. Additionally, nodes on the path add a
temporary entry for the source in the routing table. They set the
\emph{next hop }field to the identifier of the node from which they
received the packet and leave the \emph{level} and \emph{distance}
field empty.
\begin{table}[htb]
\centering
\begin{tabular}{|c|c|c|c|}
\hline
\textit{Pkt. Type} & \textit{Relay Id.} & \textit{Dest. Id.}& Success\\
$O(1)$ bits & $O(\log n)$ bits & $O(\log n)$ bits & $1$ bit\\
\hline
\end{tabular}
\caption{Probe Packet}
\label{tab:PP}
\end{table}
Upon receiving the packet, the relay node can either answer to the
source on the reverse path we just created if the answer is
negative. Alternatively, it can take action as explained in the sequel
if it has an entry for the destination.  \par We now separately detail
the beaconing and the routing algorithms underlying our routing
protocol
\subsection{Beaconing Algorithm}
In this subsection, we start by describing the first time step, when nodes have not yet
moved and no information has been exchanged. In a static network, the information exchanged in this first step would be sufficient to setup a complete routing infrastructure. On the other hand, in a mobile environment, we would need to cope with the dynamic topology and constantly update the routing tables. How we deal with a dynamic environment is explained at the end of the subsection. Let the \emph{cover radius} at level $i$, for
$i=1,...,\log\Delta$ ($\Delta$ being the diameter of the network), be
defined as $r_i=2^i$ and the \emph{flooding radius} at level $i$ be
defined as $f_i=\kappa(r_{i+1}+r_{i})$, where $\kappa$ is a constant
chosen such that $\kappa(\nu d,d)\leq \kappa d$, $\forall d$. In order for the routing algorithm to work properly, it is crucial that beacons at level $i$ are within the flooding radius of the beacons at level $i+1$, if they have common nodes inside their cover radii (see Fig. \ref{fig:floodradius}). 
\begin{figure}[htp]
\centering
\includegraphics[width=0.6\textwidth,keepaspectratio]{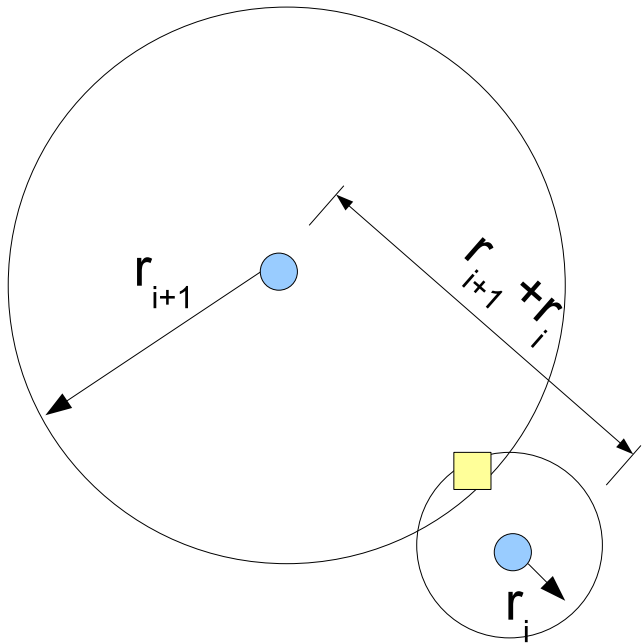}
\caption{The flooding radius is chosen in such a way that beacons at level $i$ hear the floods of beacons at level $i+1$. In static networks, this is $r_i+r_{i+1}$.}
\label{fig:floodradius}
\end{figure} 
This is why we define the flooding radius above. Note that if the network is static, we can set $\kappa$ to $1$. Else, the value of $\kappa$ depends on how mobile the nodes are (see Def. \ref{def:kappa}).
\par
The idea of the algorithm is to build a hierarchical cover of the network
\emph{i.e.,} we would like every node in the network to be within
$r_i$ hops of a beacon node at every level $i$. We say that when a node is
within $r_i$ of a beacon $b$ at level $i$, it is a member of $b$'s
\emph{cluster} at level $i$, but it can only be in one cluster at
every level. To achieve this, we let the nodes flood in a random order
which can change at every time step. Every node $u$ is a beacon at a
given level $\beta(u)$. The flooding radius, however, will depend on
the highest level at which a node is not covered. Let us denote by
$h(u)$ the highest level at which node $u$ is not covered. When node
$u$'s turn to flood comes, it will determine the value of $h(u)$ set
$\beta(u)=h(u)$ and call $flood(f_{h(u)},h(u))$. A node $v$ which
receives this flood will determine the lowest level at which it could
be a member of $u$'s cluster, say $l(v)$. That is, it will determine
the lowest value $j$ for $l(v)$ such that $d(u,v)\leq 2^j$. This
distance $d(u,v)$ is known since $v$ just received a flood packet from
$u$. It will then become a member of $u$'s cluster for all levels
above $l(v)$ for which it has no membership yet and are below
$\beta(u)$. If a node becomes a member of $u$'s cluster, it sends
a \emph{membership packet} (see Table \ref{tab:memp})
\begin{table}[htbp]
\centering
\begin{tabular}{|c|c|c|c|}
\hline
\textit{Pkt. Type} & \textit{Node Id.} & \textit{Beacon Id.} & \textit{Level}  \\
$O(1)$ bits & $O(\log n)$ bits & $O(\log n)$ bits & $O(\log \Delta)$ bits \\
\hline
\end{tabular}
\caption{Membership Packet}
\label{tab:memp}
\end{table}
 back to $u$. In this way, $u$ learns the identifier of all nodes in its cluster. Note that $u$ also applies this procedure to itself,
and consequently could be a beacon at level $i$ but not at level
$j<i$.  \par The control traffic will be dominated by the messages
sent back by nodes to beacons when they become members of a
cluster, so should be rare. Moreover, we do not want the distance between a node
and its beacon to grow by more than a constant
factor. Since we assume that the maximum speed of the node is
constant, the higher the level of a beacon, the more time it will take
for nodes to double their distance to this beacon. We want to
\emph{elect} new beacons and update memberships only for levels at
which the distances could have been multiplied by a constant factor. Recall
that the network is $\kappa$-constrained. Consequently, the distance
$\oft{d}(u,v)$ between two nodes $u$ and $v$ cannot change by a factor
$\kappa$ in less than $\nu d$ time steps (see Corollary
\ref{cor:kappa}). In particular, if a node is at distance $2^i$ of a
beacon at the time it becomes a member of its cluster, then we have
$d^{t+\nu 2^i}\leq \kappa 2^i$. Hence, we update the memberships
at level $i$ only every $\nu 2^i$ time steps (see Figure \ref{fig:freq}).
\begin{figure}[htp]
\centering
\includegraphics[width=1\textwidth,keepaspectratio]{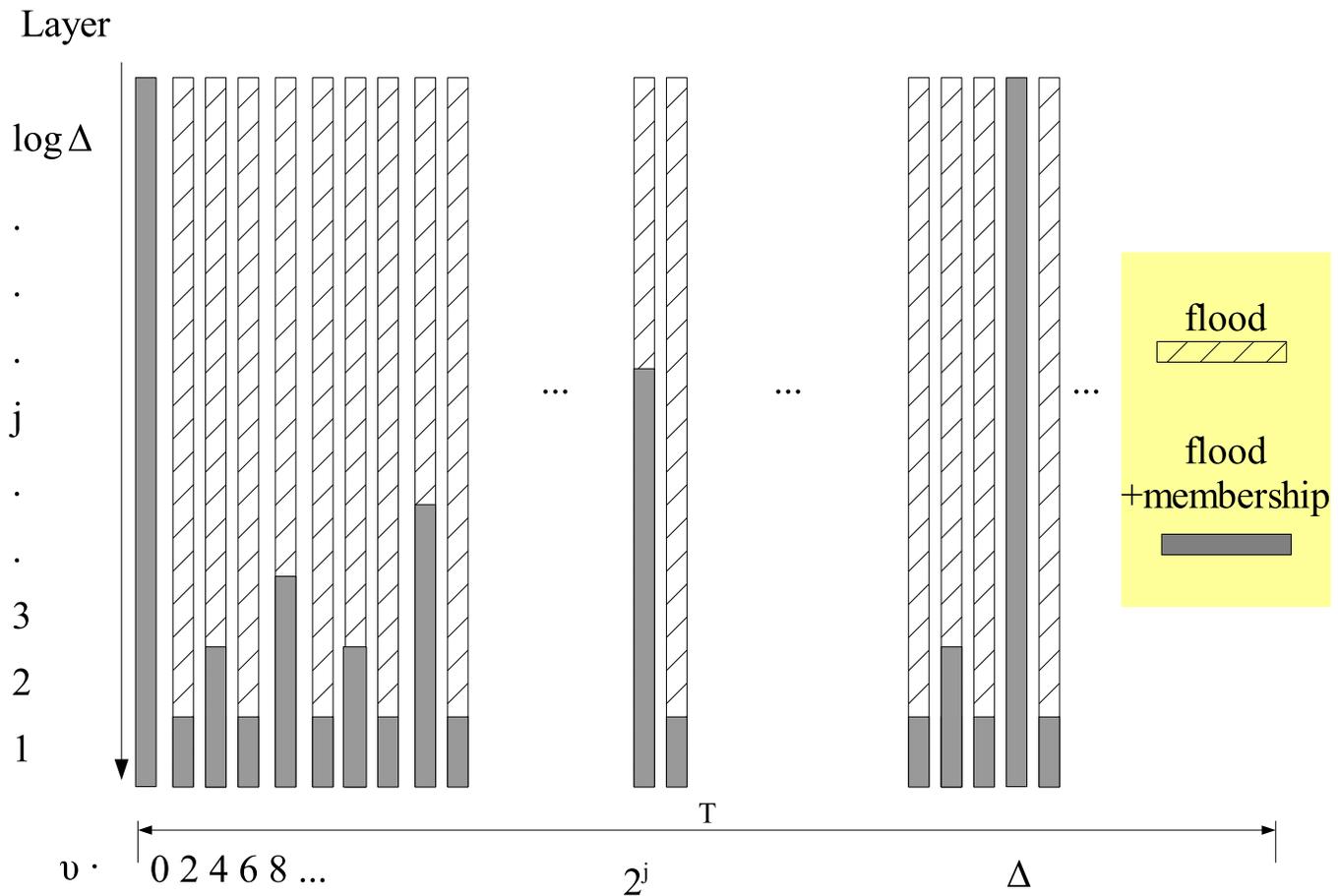}
\caption{The memberships up to level $i$ are updated every $\nu 2^i$
time steps. At the levels above, beacons elected at earlier time steps
simply flood again.}
\label{fig:freq}
\end{figure} 
This will lead to a routing scheme in which the distances can be
distorted by at most a constant factor to be calculated in the
sequel. Additionally, in a dynamic environment, routes can break. This
is why we let the beacons at all levels flood at every time
step. Levels at which no membership updates take place simply use the
floods of the beacons to update their routes toward theses
beacons. This will ensure that a route always exists for all pair of
nodes.
\begin{figure}[htp]
\centering
\includegraphics[width=1\textwidth,keepaspectratio]{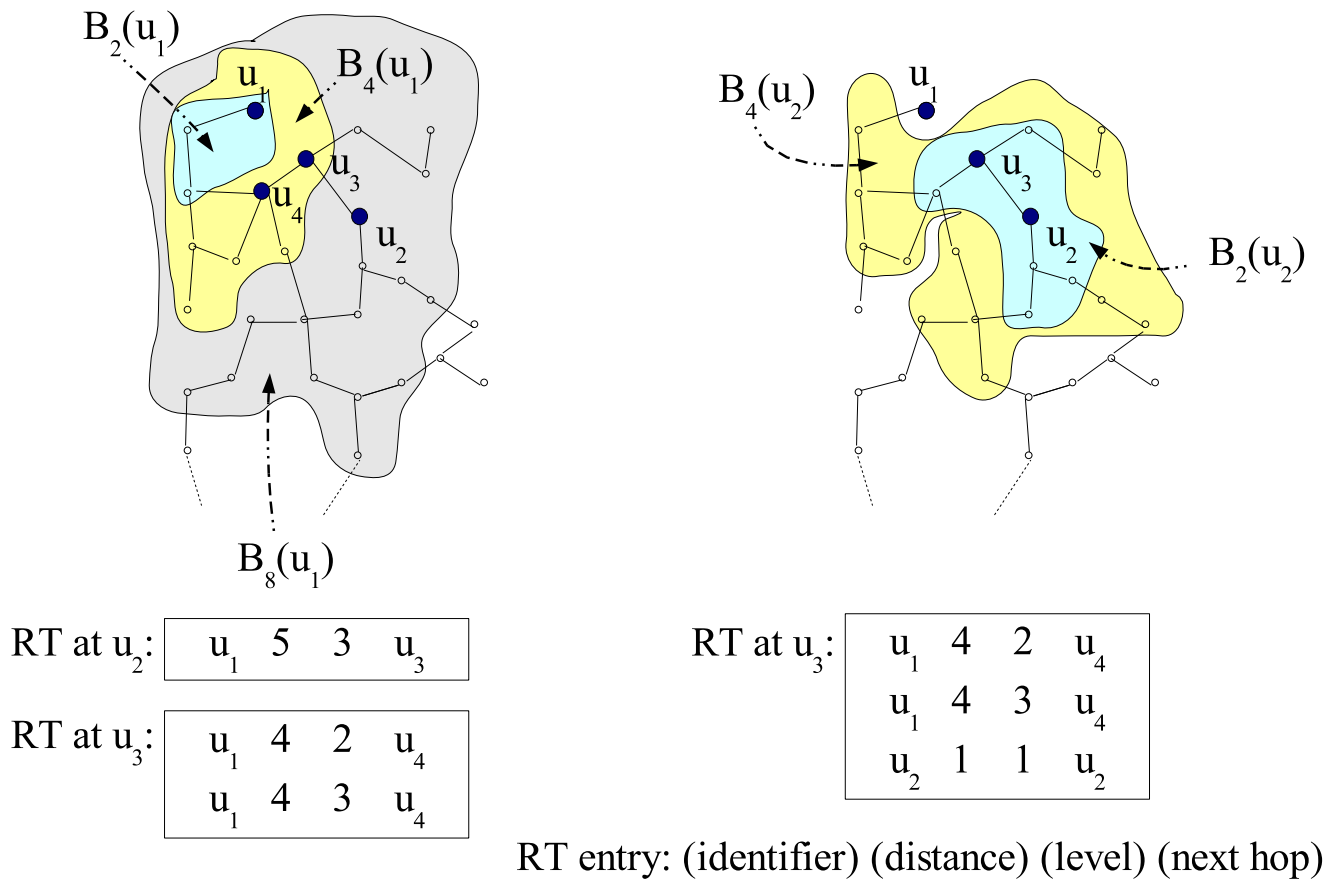}
\caption{The example start with empty routing tables. First, on the left,
node $u_1$ floods at level $3$. We focus on nodes $u_2$ and
$u_3$. Node $u_2$ is within $8$ hops from $u_1$ but further away than
$4$ hops. Consequently, it can only had an entry for node $u_1$ at
level $3$. At the same time, node $u_3$ can add an entry for node
$u_1$ at the levels $2$ and $3$, since it is at distance $4$ of
$u_1$. Next, on the right, $u_2$'s turn to flood comes
(right after $u_1$'s turn). This node is already covered at level
$3$. Consequently, it will flood at level $2$. The node $u_3$ could
potentially add an entry for this node at levels $1$ and $2$. However,
it is already covered at level $2$ and so adds only an entry for level
$1$. We do not show the entries beacons add for themselves.}
\label{fig:beac}
\end{figure}   

In Figure \ref{fig:beac}, we give a simple with three levels. The
beaconing algorithm is presented in Algorithm \ref{algo:beac}. It is
important to note that the routes are updated at every time step and
consequently routing toward a beacon will always be
successful. Further, when the membership at a given level $i$ is
updated, all the memberships at the levels $j<i$ will also be updated,
and all memberships at these levels canceled.  \small
\begin{algorithm}[htb]
\SetLine 
\KwData{Routing Table, Time t} 
\Begin{
Let $\Gamma=\max \left\{0\leq j\leq \log\Delta| t\ mod \nu 2^j=0 \right\}$\; 
Clear routing table entries with $level\leq \Gamma$\; 
Let $\beta(u)$ be the level at which $u$ is a beacon\;
\If{$\pi(\ell)=u$}
{
\If{$\beta(u)\leq \Gamma$}
{Let $h(u)$ be the highest level at which $u$ is not covered\; 
$\beta(u)=h(u)$;\
}
$flood(f_{h(u)},h(u))$\; 
}
}
\caption{Beaconing Algorithm at node $u$}
\label{algo:beac}
\end{algorithm}
\normalsize

\subsection{Forwarding Algorithm}
The forwarding algorithms works as follows: a source node $u$ with a message for a target node $v$ searches for $v$ by first
probing all the level 1 beacon it knows of. To do so, it looks at its routing table and selects all nodes it knows of at
level $1$. If all answers are negative, node $u$ probes all level $2$
beacons it knows of. The procedure is repeated as long as all beacons
answer negatively. A beacon at level $i$ with an entry for the
destination in its routing table does not answer directly to the
source. Rather, this node will search downwards in the hierarchy by probing all the level $i-1$ beacons it knows
of. We show in the next section that one of these beacons must
have an entry for the destination. That beacon in turn probes all the
beacons in knows of at level $i-2$. Meanwhile, the other beacons at
level $i-1$ will answer negatively to the beacon at level $i$. The
procedure is repeated recursively until the target itself is
reached. The target will then answer to the source on the reverse path
which will later be used for communication between the source and the
destination. To summarize, the forwarding algorithm starts with an ``upstream'' phase during which the source node probes beacons level by level until a beacon is found which has the destination in its cluster. That beacon then starte a ``downstream'' phase, during which we go down in the hierarchy. We illustrate the forwarding procedure conceptually in
Figure \ref{fig:probe}.
\begin{figure}[ht]
\centering
\includegraphics[width=1\textwidth,keepaspectratio]{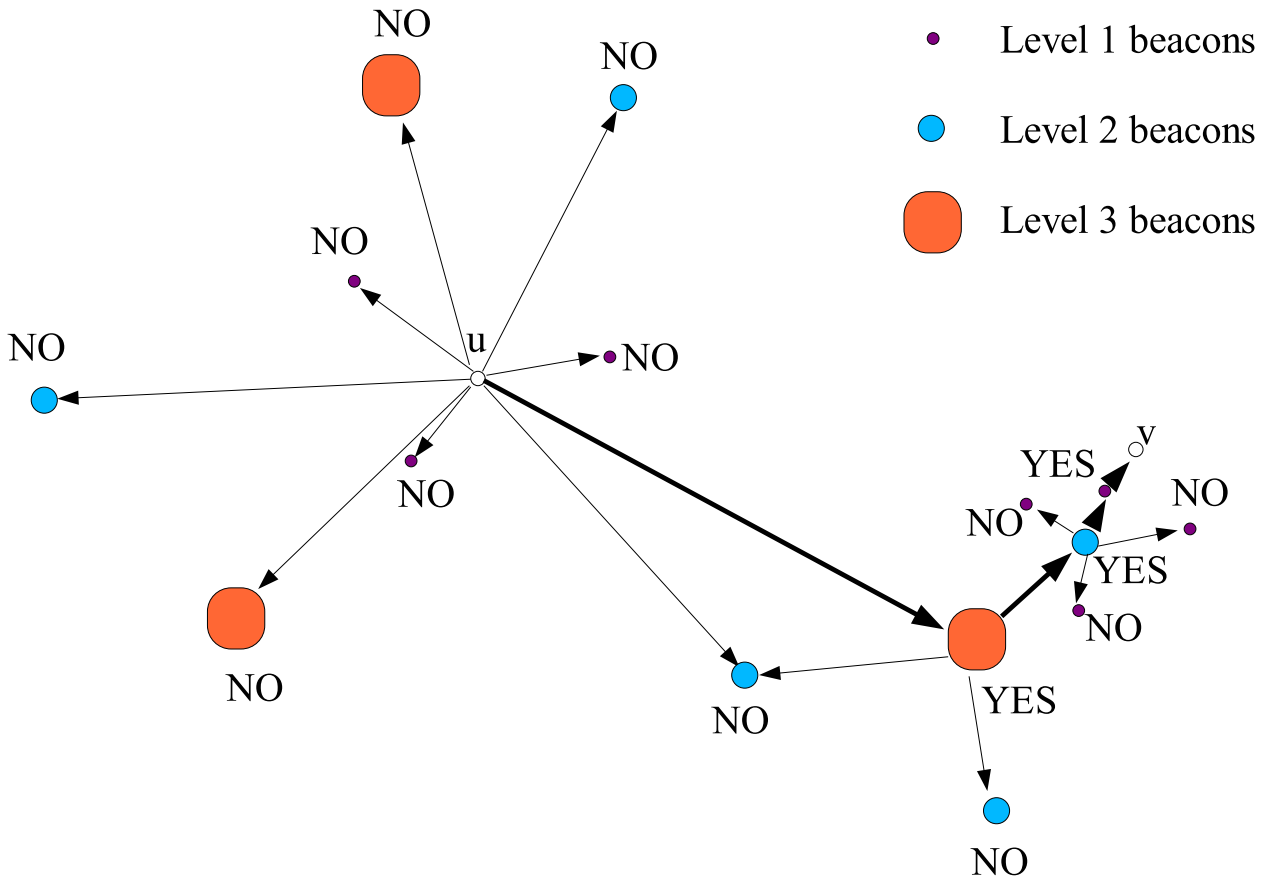}
\caption{Node $u$ has a packet for node $v$. It searches in its
routing table for all beacons it knows of at level $1$ and sends them
a probe packet containing $v$'s identifier. None of the beacons at
level $1$ has an entry for this node and consequently they all answer
negatively to node $u$. Next, node $u$ repeats the same procedure with
all the beacons it knows of at level $2$. Again, all beacons answer
negatively. On the third level, now, a beacon has an entry for node
$v$. This beacon will probe all the beacons it knows of at level $2$,
while the other beacons at level three will answer negatively to
$u$. A beacon at level $2$ must have an entry for $v$. This beacon
again probes all the beacons it knows of at level $1$ among which one
must have an entry for $v$ itself. Meanwhile, the other beacons reply
negatively as they do not have any entry for $v$.}
\label{fig:probe}
\end{figure}   
\subsection{Load-balancing}
This approach above guarantees a low network wide control traffic
overhead. Even though over a long period of time all nodes will get
approximately the same average overhead, beacons at the highest levels
might get overloaded by the membership packets of the nodes in their
cluster when a membership update takes place. These nodes will be hot
spots in the network for a short period of time. To work around this
problem, memberships can be distributed in the cluster instead of
stored at the beacon itself. First, we now set
$f_i'=\kappa(2r_{i+1}+r_i)$. Additionally, whenever a beacon floods at
level $i$, it includes its membership at level $i+1$ in the
packet. This information is stored by all nodes that receive this
flood packet. This will guarantee that all nodes that are members of a
cluster at level $i$, know how to reach all beacons at level $i-1$
inside that cluster. A node that becomes a member of the cluster of
beacon $b_i(u)$ at level $i$ will now send its membership packet
directly toward the beacon $\psi_{i-1}(u)$ at level $i-1$ inside this
cluster with the identifier closest to $u$'s. In turn, as soon as the
packet reaches a node which is a member of $\psi_{i-1}(u)$'s cluster
at level $i-1$, the membership packet is redirected toward the beacon
$\psi_{i-2}(u)$ which is a member of $\psi_{i-1}(u)$'s cluster at
level $i-1$ and has the identifier closest to node $u$'s. The process
is repeated until we reach a single node, which will store $u$'s
identifier on behalf of $b_i(u)$. Note that the membership can only be
registered at a single location in the cluster reachable through a
unique sequence of clusters. This remains true even when nodes
move. Indeed, the nodes in the cluster of $b_i(u)$ will only forward
the packet to beacons at level $i-1$ which were in the same cluster at
the time the membership for this level got updated. Of course,
whenever level $j<i$ is updated, we do now not only need to send
$u$'s identifier toward its new beacon at that level. Additionally,
the node that holds $u$'s identifier at level $j$ might not be
reachable anymore through a path of clusters with identifiers closest to
$u$'s. Consequently, this node will need to forward $u$'s identifier
toward the beacon at level $j$ with the identifier closest to
$u$'s. Again the process will be repeated recursively until a single
node is reached. As we will see, the cost of avoiding hot spots is a
factor $\log{n}$ in the total control traffic. Finally and most
importantly, with this procedure beacons no longer get
overloaded. Rather, the traffic is be distributed in its cluster.
\par The data forwarding process remains the same except that the
source node will not probe the beacon itself, but rather search for
the node in the beacon's cluster that should hold the destination's
identifier. If this node holds the identifier, it will then probe the
beacons one level below in the same way. Recall that the nodes which
potentially hold $u$'s membership can be reached at any given instant
in time through a unique sequence of clusters. The procedure is
repeated until the destination is reached.

\section{Performance Analysis}
\label{sec:PerfAnal}

In this section, we analyze the performance of our algorithm
analytically both in terms of control traffic and of route stretch. As
in Section \ref{sec:RoutingAlgo}, we do this for a sequence of
doubling and smooth connectivity graphs, and will use
$\mathcal{G}^{(t)}(n,r_n)$ with USL mobility for illustration.

The bounds derived in this section hold w.h.p. when we are in a
sequence of length
$n^{\rho}$ of $\alpha$-doubling connectivity graphs. In the sequel, $\alpha$, $\kappa$ and $\nu$ are the
constants derived in Section \ref{sec:npp}. Let us denote by
$\Delta=O(\sqrt(\frac{n}{\log(n)}))$ the diameter of the network. To
bound the control traffic necessary for beaconing, we will rely on the
$\alpha$-doubling property of the metric space to show that a node can
only hear a constant number of beacons at every layer. We will first
show that a ball of radius $2R$ around any node $u$ can only contain
at constant number of balls (clusters) of radius $R$, when we select
the centers of the balls of radius $R$ in an arbitrary order and
ensure that two centers cannot be closer than $R$. We will later use
this result to show that a node can hear at most a constant number of
beacons at any given level.
\begin{theorem}[Random Cover]
\label{TRC}
Let $\B{X}{u}{2R}$ be a ball of radius $2R$ centered at $u$ in a graph metric $(X,d)$
with doubling constant $\alpha$. Then, there exist
at most $k\leq \alpha^2$ nodes $v_i$, $(i=1,2,..,k)$ such that
$\B{X}{u}{2R}\subseteq \bigcup_{i}^{k}\B{X}{v_i}{R}$ and
$min_{(i,j)}d(v_i,v_k)>R$.
\end{theorem}
\begin{proof}
By definition of an $\alpha$-doubling metric space, there must exist a
cover of a ball of radius $2R$ consisting of at most $\alpha$ balls of
radius $R$. Recursively, there must also exist an $\frac{R}{2}$-cover
consisting of $\alpha^2$ points. One can select at most one center
$v_i$ in each ball of radius $R/2$, as any other point inside this
ball is within $R$ of $v_i$. Hence, one can select at most $\alpha^2$
such centers.
\end{proof}
%One can easily extend this result to the cover of larger balls.
\begin{corollary}
\label{LRC}
Let $B$ be a ball of radius $R>R'$ in an $\alpha$-doubling metric
space $(X,d)$. Then, one can select at most $k\leq
(\frac{R}{R'})^{2log(\alpha)}$ nodes $v_i$, $(i=1,2,..,k)$ such that
$\B{X}{u}{R}\subseteq \bigcup_{i}^{k}\B{X}{v_i}{R'}$ and
$min_{(i,j)}d(v_i,v_j)>R'$. In particular, if $R=\eta R'$ for some
constant $R$, then $k$ is at most a constant $(\eta)^{2log(\alpha)}$
independent of $n$.
\end{corollary}
\begin{proof}
Let
$R=2^iR'$. Hence, $R'$ is doubled $log\frac{R}{R'}$ times to obtain
$R$. By Theorem \ref{TRC}, $B$ can be covered by
$\alpha^{2log\frac{R}{R'}}=(\frac{R}{R'})^{2log(\alpha)}$ balls of
radius $R'$.
\end{proof}
Here, one can think of the radius $R$ of the large balls as the
flooding radius, and of the radius $R'$ of the small balls as the
cover radius. Indeed, we use this result to show that a node $u$ can
hear the floods of all beacons within a given radius $R$. Moreover,
this ball of radius $R$ can contain at most
$(\frac{R}{R'})^{2log(\alpha)}$ beacons, since beacons must be at
least $R'$ apart.
\subsection{Control Traffic}
%In this section we will use the intuition we developed above to bound the control traffic.
\begin{theorem}
\label{thm:cto}
The average control traffic overhead per time step for beaconing is at most 
$O(n\log^2n)$ bits.
\end{theorem}
\begin{proof}
We will analyze the control traffic at level $i$. Recall that a beacon
at level $i$ floods a distance $f_i=\kappa (2^{i+1}+2^{i})$ at every
time step. Further, at the time the memberships are updated at level
$i$, a beacon node at this level \emph{cannot} be within $r_i=2^i$ of
another beacon at that level. If it were the case, this node would not
elect itself as a beacon at this level. Level $i$ is updated every
$\nu 2^i$ time steps. Consider a node $u$. By Corollary \ref{cor:rwk},
no nodes that are further away than $\kappa f_i$ hops at the time the memberships
are updated at level $i$ could move within $f_i$ of $u$ in less than
$\nu 2^i$ time steps. However, that is before this level is updated
again. Consequently, the number of beacons whose flood can reach $u$
at any given time step is at most the number of level $i$ beacons in a
ball of radius $\kappa f_i$ at the time the membership is updated. In
turn, node $u$ will broadcast\footnote{recall that when a node
broadcasts a packet it is received by all direct neighbors in the
connectivity graph. Consequently, there is one packet transmission per
beacon of which a flood packet is received.} the flood packets of at
most that many beacons for this level $i$. By Corollary \ref{LRC},
this number is a constant\footnote{In the load-balanced scheme, this
constant is $(5\kappa^2)^{2\log\alpha}$.} given by $(\frac{\kappa
f_i}{2^i})^{2log(\alpha)}=(3\kappa^2)^{2\log\alpha}$. Given that there
are $O(\log n)$ levels, that there are $n$ nodes and that a flood
packet has size $O(\log n)$ bits, the average control traffic overhead
per time step for beaconing is at most $O(n\log^2n)$ bits.
\end{proof}
We now compute the control traffic overhead necessary for nodes to
update their memberships with beacons. Recall that level $i$ and all
levels below are updated every $\nu 2^i$ times steps and that a node
can only be a member of one cluster at every level. Furthermore, a
node only becomes a member of a cluster if it is within $2^i$ of the
corresponding beacon.
\begin{theorem}[Membership Update Overhead]
The average control traffic overhead per time step to update
memberships without load-balancing is at most \[\frac{n\log\Delta\log
n}{\nu}=O(n\log^2n)\] bits.
\end{theorem}
\begin{proof}
Consider a sequence of $T$ time steps. The memberships
will be updated up to level $i$ every $\nu 2^i$ time steps, so
$\frac{T}{\nu 2^i}$ times in a sequence of length $T$. At the time of
the update, a node can be at distance at most $2^i$ from a beacon at
level $i$. Consequently, the overhead in bits generated by a node in a
sequence of $T$ time steps is upper bounded by
$\sum_{i=1}^{\log\Delta}\frac{T}{\nu 2^i}2^i\log
n=\frac{\log\Delta}{\nu}\log n$.
\end{proof}
Finally, we will show that the average control traffic overhead when
load-balancing is used is increased by at most a factor $\log n$.
\begin{theorem}[Membership Update Overhead]
The average control traffic overhead per time step to update
memberships with load-balancing is at most
\[
\frac{n\log^2\Delta\log n}{\nu}=O(n\log^3n)
\]
bits.
\end{theorem}
\begin{proof}
Consider a sequence of $T$ time steps. The memberships will be updated
up to level $i$ every $\nu 2^i$ time steps, so $\frac{T}{\nu 2^i}$
times in a sequence of length $T$. At the time of the update, a node
can be at distance at most $2^{i+1}$ from a beacon at level $i-1$
inside its cluster at level $i$. Similarly, a node can be at distance
at most $2^{i}$ from a beacon at level $i-2$ inside its cluster at
level $i-1$. In the load balanced scheme, we have to count the
overhead to go down the hierarchy of beacons. For a beacon at level
$i$, this is at most $2^i\times 2$. Consequently, the overhead in bits
generated by a node in a sequence of $T$ time steps is upper bounded
by $4\sum_{i=1}^{\log\Delta}\frac{T}{\nu
2^i}2^i\log n=4\frac{\log\Delta}{\nu}\log n$. However, node $u$ is a member
of a cluster at all $\log{\Delta}$ levels. Recall that the node that
holds $u$'s identifier must always be reachable through a path by
choosing the beacon (cluster) with the identifier closest to
$u$'s. Hence, whenever level $i$ gets updated, all $\log{\Delta}$
nodes that hold $u$'s identity must follow the same procedure as $u$
itself. We conclude that the overhead is upper bounded by
$\log{\Delta}4\frac{\log\Delta}{\nu}\log n$ bits.
\end{proof}

\subsection{Route Stretch}
In this section we will show that the route found with the forwarding
algorithm is only a constant factor longer than the shortest
path. Additionally we show that the destination location discovery
takes a negligible fraction of a flow throughput.

\begin{theorem}[Routing Stretch]
The worst case multiplicative routing stretch is $O(1)$. 
\end{theorem}
\begin{proof}
We first analyze the stretch without load balancing. Consider that we
want to route from a node $u$ to a node $v$, and that we had $2^k\leq
d(u,v)\leq 2^{k+1}$, the last time level $k$ was updated before the
route search takes place. Let us denote by $b_i(v)$ the beacon to
which node $v$ had registered the last time level $i\leq k$ was
updated before the route search takes place. Clearly, we have
$d(u,b_v(k))\leq\kappa (2^{k+1}+2^{k})$, and $d(b_i(v),b_{i-1}(v))\leq
\kappa (2^i+2^{i-1})$. This is true since the membership of node $v$
at level $i$ must have been updated at most $\nu 2^{i}$ time steps
before the routing takes place, and that at the time the time level
$i$ gets updated, we have $d(b_i(v),b_{i-1}(v)))\leq
d(b_i(v),v)+d(v,b_{i-1}(v))$ by triangle inequality. Note that
$d(b_i(v),b_{i-1}(v))\leq f_{i-1}$ and that $d(u,b_v(k))\leq
f_k$. Hence, a route must exist between $u$ and $v$ and the length
$r^{(t+\tau)}(u,v)$ of the route at time $t$ is at most:
\[
\begin{array}{ll}
r^{(t)}(u,v)&\leq \sum_{i=1}^{k}f_k
=\kappa\sum_{i=1}^{k}(2^{i+1}+2^i)\\
&=3\kappa\sum_{i=1}^{k}2^i=3\kappa \frac{2^{k+1}-1}{2-1}
\leq 6\kappa d^{(t)}(u,v) 
\end{array}
\] 
In the worst cast, nodes $u$ and $v$ have moved closer together (by a
factor $\kappa$) while the beacons have moved further apart. Indeed,
we have $d^{(t+\tau)}(u,v)\leq \kappa d^{(t)}(u,v)$ for $\tau\leq\nu
2^k$ as our network is $\kappa$-constrained. Note that if we waited
longer that $\nu 2^k$, memberships would be updated again at level $k$
and we could find another beacon at distance $2^k$ at most from $v$ at
level $k$. Hence, the worst case stretch is $\frac{r^{(t+\nu
2^k)}(u,v)}{d^{(t+\nu 2^k)}}\leq 6\kappa^2=O(1)$.  
\end{proof}

Every node can only hear floods from a constant number
($\mu=(3\kappa^2)^{2\log\alpha}$, see Theorem \ref{thm:cto}) of
beacons at every level. Recall that the source will first probe all
beacons at level $1$, then all beacons at level $2$ and so on. The
procedure is repeated up to level $k$ at which the source $u$ will
send a packet to $b_k(v)$. Note that the distance from $u$ to this
beacon can be at most $\kappa 2^{k+1}+2^{k}=f_k$ and so it must hear
its floods. In turn, when routing down the hierarchy, beacon $b_j(v)$
will probe at most a constant number ($(3\kappa^2)^{2\log\alpha}$ of
beacons at level $j-1$. Finally, the distance between a node $u$ and a
beacon at level $i$ can be at most $f_i$ and a probe packet will
traverse at most $2f_i$ packets when a beacon at level $i$ is probed
(back and forth). This means that for discovery of the location of the
destination, we need a probe overhead of at most $\mu 6\kappa d(u,v)$
packet transmissions. Therefore, this is a negligible part of the
throughput of a flow since it consumes roughly the equivalent of a few
packet headers of a flow from source to destination.  A similar
statement can be made for the load-balanced case.

\iffalse
\[
\begin{array}{ll}
\mbox{overhead}&\leq 2\mu\sum_{i=1}^{k}f_k\\
&\leq 2\mu\sum_{i=1}^{k}2^i(3)\\
&\leq \mu 6\kappa d(u,v) \\
\end{array}
\]

Finally, in the load-balanced case, the
route can be at most 2 times longer than without
load-balancing. Indeed, a node $\psi_{k-1}(u)$ (in this case a beacon
at level $k-1$) inside $b_v(k)$'s cluster at the time the membership
gets updated at this level can be at distance at most
$d(u,\psi_{k-1}(u))\leq\kappa (2^{k+1}+2*2^{k})\leq \kappa
2((2^{k+1}+2^{k}))$ from $u$ at any time instant before the membership
is updated again. The same is true at the levels below.
\fi

\section{Implementation Issues}
\label{sec:impl}

In Section \ref{sec:PerfAnal}, we have computed worst case bounds
which may be conservative in terms of constants.  In this section, we
explore this aspect by looking at simulation results for the
control traffic and for the stretch. Recall that we had computed that
for each of the $O(\log n)$ levels, a node has to retransmit a packet
of at most $(3\kappa^2)^{2\log\alpha}$ beacons. Even if we set the
maximum speed as well as the parameter $\nu$ to $1$, this is still
$\sqrt{10}+20$ and consequently the constant in the bound on the
overhead at least as high as $(3(\sqrt{10}+20)^2)^2\approx 2.5\cdot
10^6$!  In Figure \ref{fig:simoh}, however, we show that in practice
this constant is approximately 30. This simulation was run with $50$
up to $10000$ nodes moving at a maximum speed of $1$. One can observe
that the experimental scaling behavior corresponds extremely well to
the theoretical behavior. To stress this fact, we also plot $100\log
n$ as a benchmark. Note that the overhead is expressed in number of
packets rather than bits (a packet being of size $O(\log n)$).
\begin{figure}[htbp]
\centering
\includegraphics[width=0.9\textwidth,keepaspectratio]{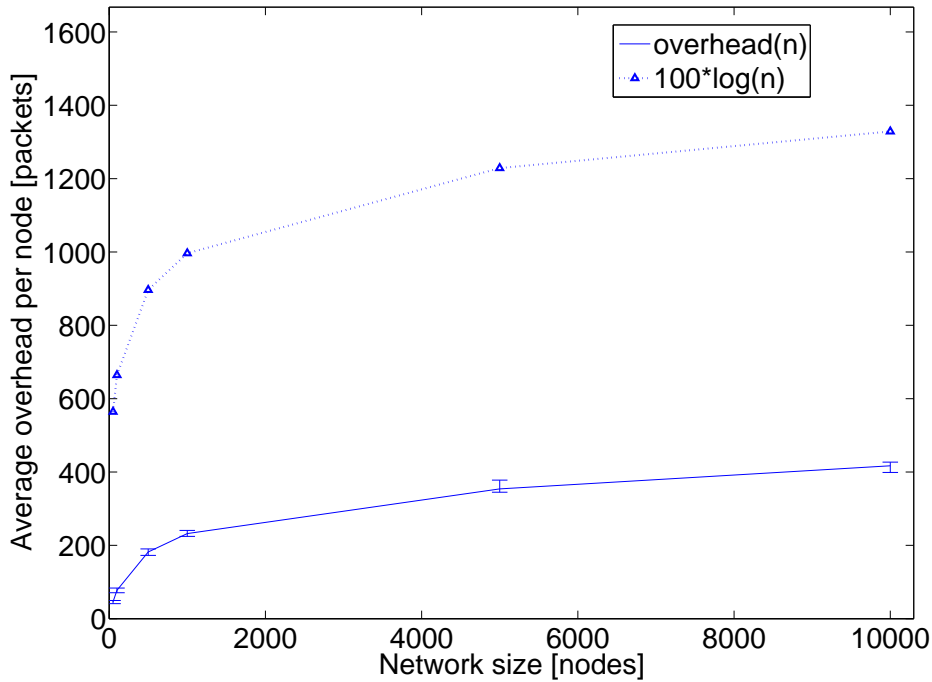}
\caption{Average control traffic overhead per node in packets as a
function of the network size. Nodes move at a speed of maximum speed
of $1$. The confidence interval is given by the 95\% and 5\%
percentiles. The size of a packet is $O(\log n)$ bits. We also plot
$100\log{n}$ to show that our analytical predictions match the
simulation results.}
\label{fig:simoh}
\end{figure}  

Similarly, in Figure \ref{fig:simstr} we show that for a network of
$1000$ nodes, the stretch is at most $1.5$ for all node pairs. If we
compute the maximum theoretical stretch, we can show that it is again
considerably larger and hence a pessimistic bound.  These small
constants could make a practical implementation realistic.
\begin{figure}[htbp]
\centering
\includegraphics[width=0.9\textwidth,keepaspectratio]{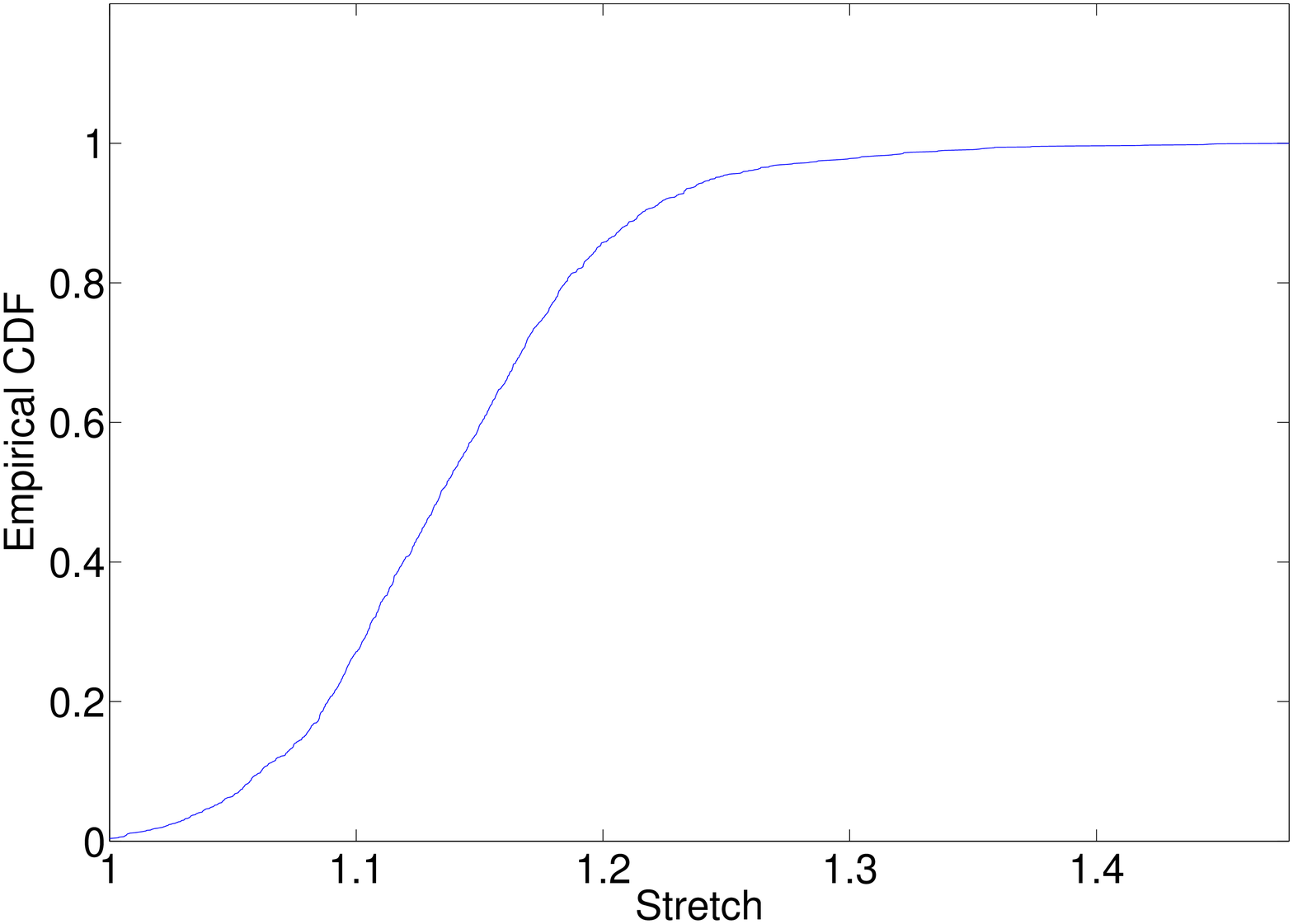}
\caption{Empirical cumulative distribution of stretch (route
length/shortest path) for a network with $1000$ nodes moving at a
maximum speed of $1$.}
\label{fig:simstr}
\end{figure}

We have made a certain number of assumptions in our models, which we
now clarify. In practice, the random permutations on the nodes,
which determines the order in which the flooding occurs, could be
implemented by using random timers; more precisely, by letting all
nodes draw a random delay independently of each other every $\Delta T$
seconds. Obviously, the interval from which nodes draw this delay
should be made sufficiently large so that we can avoid
collisions. However, a level in the hierarchy will be rapidly covered,
and in a practical implementation the covers at different levels could
be built in parallel. Further, different parts of the network are
independent except at the highest level, and we could exploit this
spatial diversity to parallelize the beaconing process. Hence, we speculate
that it is possible to reduce the length of the beaconing phase to a
small constant times the maximum round-trip time. Note that one could
apply the algorithms to underlying networks that are not doubling. In
this case, we would not be able to give provable bounds on the control
overhead and the stretch as we did for doubling networks.

\section{Conclusions}
In this paper, we show that a large class of wireless network models
belong to a larger class of networks, the \emph{doubling} networks, in
which efficient routing can be achieved. To design an efficient
routing scheme, one can hierarchically decompose the network by
relying on the doubling property to prove that the control traffic
overhead and the stretch will remain low, even for dynamic doubling
networks.  This holds for a fairly broad class of uniform
speed-limited (USL) mobility models.  One advantage of the proposed
routing algorithm is that it is robust, in that it works well in certain situations in which
other existing algorithms cannot work well. This was illustrated in
Section \ref{subsec:GnR} for an example network with obstacles.  We
believe that many more such examples can be created where the use of
the doubling rather than geographic properties would be crucial.  To
the best of our knowledge, our results are the first provable bounds
for routing quality and costs for dynamic wireless networks.  These
techniques might give us insight into algorithm design for more
sophisticated wireless network models.

\bibliography{doubling}  % sigproc.bib is the name of the Bibliography in this case
\pagebreak
%\clearpage
\appendix
%Appendix A
\paragraph{Unit Disc Graphs}
\label{apx:A}
Another common model used in studies on wireless networks are \textsl{Unit Disk Graphs} (UDG), which are the deterministic variants of the random geometric graphs. The randomness of the positions of the nodes is removed and they can be placed arbitrarily on a finite of infinite area. The channel model is completely deterministic as before and nodes are connected if their Euclidean distance is below a threshold distance $r$, called the communication radius. In mathematical terms, two nodes $u$ and $v$ with positions $x(u),x(v) \in [0,R]^2$ are connected if and only if $||x(u)-x(v)||<r$. We will now show that there exist UDG which are not $\alpha$-doubling (see Section \ref{sec:mod} for a definition of an $\alpha$-doubling metric).
\begin{theorem}
\label{thm:udg}
There exists an infinite UDG for which is no constant that upper bounds the doubling dimension \emph{i.e.,} UDG are not doubling.
\end{theorem}
\begin{proof}
Consider the graph shown in Figure \ref{UDG}. 
\begin{figure}[htbp]
	\centering
	\includegraphics[width=1\textwidth,keepaspectratio]{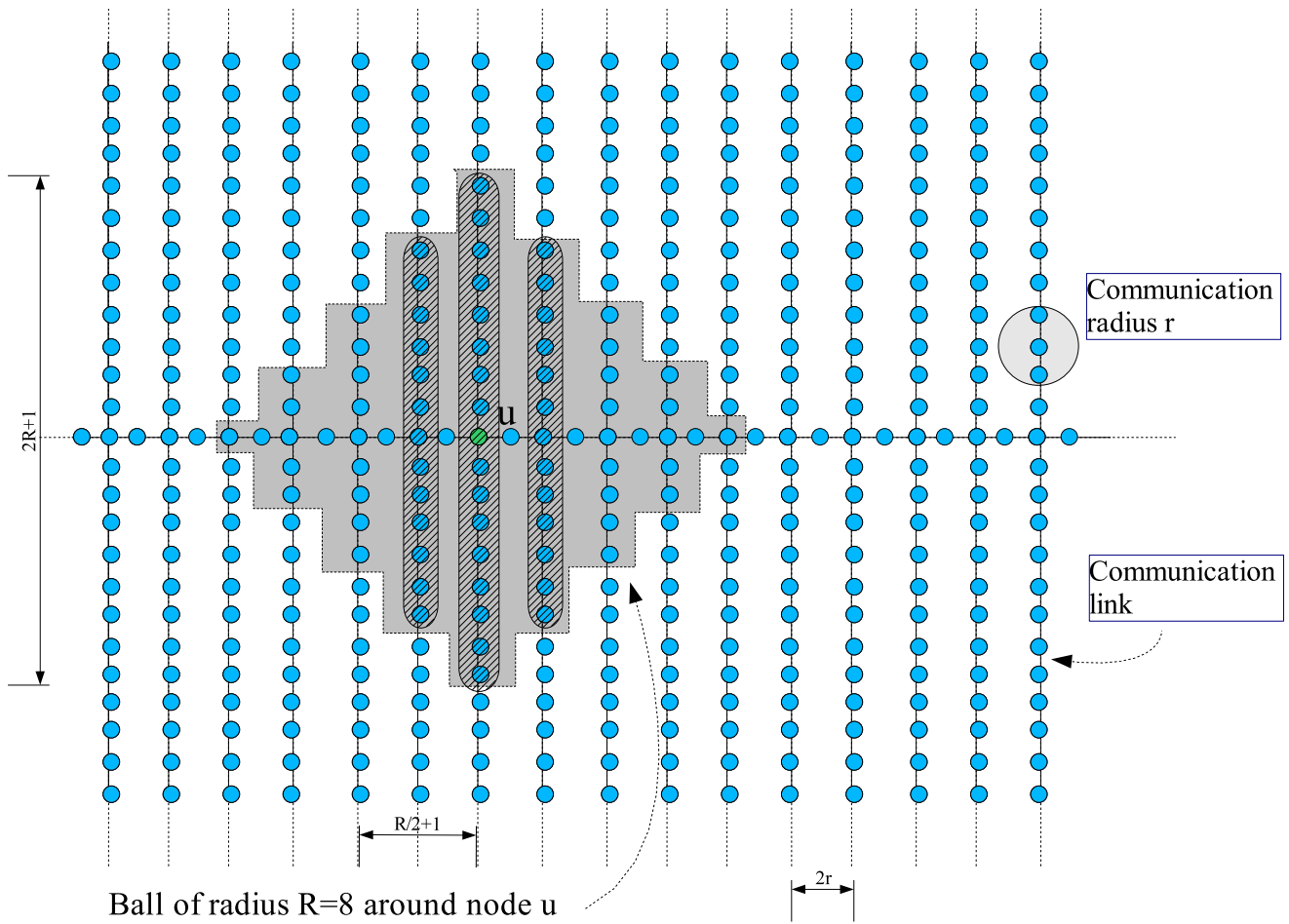}
	\caption{An infinite UDG obtained by deleting all the nodes in every second column of a grid, except for the nodes on the the middle row. Consequently, ``columns'' are $2r$ apart.}
	\label{UDG}
\end{figure}
To show that this graph is not $\alpha$-doubling, we must show that there exists no constant such that all balls of radius $R$ can be covered a constant $\alpha$ number of balls of radius $R/2$, for all $R$. Consider the ball centered around $u$ in the figure. One can see that there are $R/4+1$ ``columns'' which cross the middle row at a distance less than $R/2$ from $u$ (that is, the intersection of the column and the row is less than $R/2$ hops away from $u$). The intersection of each of these columns with $B_u(R)$ is of length more than $R$ (see hatched zones on Figure \ref{UDG}). Consequently, for each of these columns there is at least one node at distance more than $R/2$ from the middle row. To cover these nodes, we need to place at least one ball of radius $R/2$ on each of these columns. Hence, the doubling dimension is lower bounded by $R/4$ and tends to infinity as $R$ goes to infinity.  
\end{proof}
One can notice that in the non-doubling UDG in the proof of Theorem \ref{thm:udg} results from a careful construction. In Appendix \ref{apx:B}, we show however that such a structure will occur with high probability when $r_n<\sqrt{\log n}$ in random geometric graphs.

\paragraph{Random Geometric Graphs with $r_n<\sqrt{\log n}$}
\label{apx:B}
We first consider the case in which the communication radius $r$ is such that $r_n=r=(\log{n})^{\frac{1}{2}-\frac{\theta}{2}}<\log^{1/2}n$ and $\theta\in \left.\right]\zeta,1\left.\right]$. $\zeta$ is a constant such that $0<\zeta<1$. 
\begin{lemma}
\label{lemma:wg}
For any constant $\beta$, there exists constants $\gamma>0$ and $b>0$ such that a small square area of side $\gamma r$ with $b$ nodes contains a subgraph of doubling dimension $\beta+1$ with probability $p>0$.
\end{lemma}
\begin{proof}
\begin{figure}[htbp]
	\centering
	\includegraphics[width=1\textwidth,keepaspectratio]{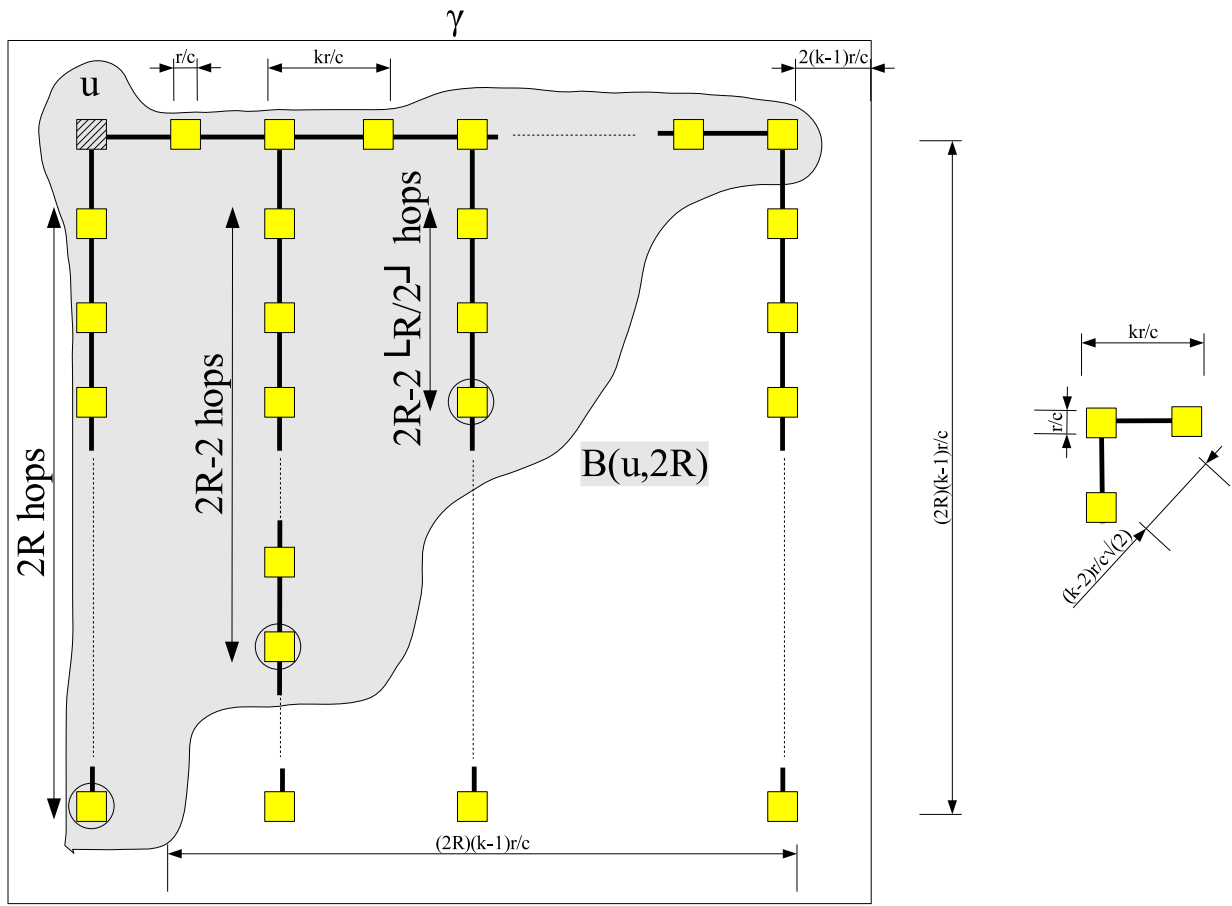}
	\caption{Graph for the proof of lemma \ref{lemma:wg}}
	\label{fig:noDoubleRand}
\end{figure}   
Consider the small square shown in Fig. \ref{fig:noDoubleRand} of side $\gamma r$, where $\gamma$ is a constant independent of $n$ to be specified later. Subdivide the small square further into mini-squares of side $r/c$. Choose the constant $c$ such that there exists a constant $k$ satisfying $\sqrt{2}(k-2)>c\geq\sqrt{k^2+1}$. Under these conditions, two nodes in mini-squares separated by $(k-2)$ other mini-square will be connected, but not mini-squares $r(k-2)\sqrt{2}/c$ apart (see right hand side of Fig. \ref{fig:noDoubleRand}). Consider now the graph on the left hand side of Fig. \ref{fig:noDoubleRand}. Assume that each full (colored) mini-square contains exactly one node. We now focus on the ball $\B{\mathcal{G}}{u}{2R}$ and will lower bound the number of balls of radius $R$ necessary to cover it. On the $\floor{R/2}$ first vertical branches from the left, the last node of the branch inside that ball (circled) must be covered by a ball of radius $R$ centered on the same branch. This is clear since the length of the branch is larger than $R$. Consequently, the doubling dimension of this graph is at least $\floor{R/2}\geq\frac{R-1}{2}$. We want the doubling dimension to be larger than $\beta$, which can be easily achieved by choosing $R$ such that $\frac{R-1}{2}>\beta$. Let $R=2\beta+2>2\beta+1$. Further, we can now set $\gamma=(2R+5)(k-1)/c=(4\beta+9)(k-1)/c$ and $b=(2R+1)\ceil{\frac{2R+1}{2}+1}=(4\beta+5)\ceil{\frac{4\beta+5}{2}+1}$. This ensures that the doubling dimension is strictly larger than $\beta$.
\par
It remains to be shown that when such a small square contains $b$ nodes, the graph constructed above occurs with probability $p>0$. The number $m$ of mini-squares contained in a small square of side $\gamma r$ is $m=\frac{\gamma^2 r^2}{r^2/c^2}=\gamma^2 c^2$ which is constant. Each node can fall in any of the $m$ squares with equal probability. Hence, all $m^b$ configurations are equiprobable and $p=\frac{1}{m^b}>0$. 
\end{proof}
We number the small squares from $1$ to $m=\frac{n}{(\gamma r)^2}=\frac{n}{\gamma^2 \log^{1-\theta}{n}}$ and denote by $X_i^{b}$ the indicator variable that takes value $1$ when small square $i$ contains exactly $b$ nodes. 
\begin{lemma}
\label{lemma:los}
There are at least $n^{1/2}$ squares containing $b$ nodes with probability at least $(1-O(\frac{1}{e^{n^{0.25}}}))$ for n sufficiently large
\end{lemma}
\begin{proof}
\[
\begin{array}{ll}
\E{X}&=\E{\sum_{i=1}^{m}X_i^b}\\
&=\sum_{i=1}^{m}\pr{X_i^b}\\
&=\sum_{i=1}^{m}{n\choose b}(\frac{1}{m})^b(1-\frac{1}{m})^{n-b}\\
&\geq m (\frac{n}{b})^b(\frac{1}{m})^b(1-\frac{1}{m})^n\\
&\geq \frac{n}{b^b}(\gamma^2\log^{1-\theta}{n})^{b-1}(1-\frac{1}{m})^{m\gamma^2\log^{1-\theta}{n}}\\
&\geq \frac{n}{b^b}(\gamma^2\log^{1-\theta}{n})^{b-1}\frac{1}{2^{2\gamma^2\log_2^{1-\theta}{n}/\log_2^{1-\theta}{e}}}\\
&\geq O(n^{1-O(\frac{1}{\log^{\theta}{n}})})\\
&\geq O(n^{\delta})
\end{array}
\]
where $\delta\geq\frac{7}{8}$ for $n$ sufficiently large, since $\theta>\zeta$.\par
Let $S_i$ be the random variable representing the small square into which the $i^{th}$ node falls. Let $F$ be the number of small squares containing exactly $b$ nodes after all nodes have been placed. Then the sequence $Z_i=\E{F|S_1,...,S_i}$ is a Doob Martingale. One can show that $F=f(S_1,S_2,...,S_n)$ satisfies the Lipschitz condition with bound $1$. Indeed, changing the placement of the $i^{th}$ ball can only modify the value of $F$ by at most $1$. We therefore obtain:
\[
\pr{|F-\E{F}|\geq n^{5/8}}\leq 2e^{-2n^{10/8-1}}=2\frac{1}{e^{2n^{1/4}}}
\]
by the Azuma-Hoeffding inequality. Consequently,
\[
\begin{array}{ll}
\pr{F<n^{1/2}}&<\pr{F<\underbrace{\E{F}-n^{5/8}}_{=n^{7/8}-n^{5/8}>n^{1/2}}}\\
&\leq2\frac{1}{e^{2n^{1/4}}}\leq 2\frac{1}{e^{n^{1/4}}}
\end{array}
\]
and 
\[
\pr{F\geq n^{1/2}}\geq (1-2\frac{1}{e^{n^{1/4}}})
\]
\end{proof}
It now remains to show that in this regime, \G{n}{r} are not doubling with high probability.
\begin{theorem}
\G{n}{(\log{n})^{\frac{1}{2}-\frac{\theta}{2}}}, where $\theta\in \left.\right]\zeta,1\left[\right.$ and $\zeta$ is a constant such that $0<\zeta<1$, are not \emph{doubling} with high probability. 
\end{theorem}
\begin{proof}
By Lemma \ref{lemma:wg}, for any constant $\beta$, a small square area of side $\gamma r$ with $b$ nodes contains a graph of doubling dimension $>\beta$ with probability $p>0$. By Lemma \ref{lemma:los}, there are $n^{1/2}$ such small squares containing $b$ nodes w.h.p. Let $F$ denote the number of small squares containing exactly $b$ nodes. Consequently, the probability that at least one of this squares contains a graph of doubling dimension $>\beta$ is given by:
\[
\begin{array}{ll}
\pr{\mbox{not doubling}}&=\sum_{j=1}^{m}\pr{\mbox{not doubling}|F=j}\pr{F=j}\\
&\geq (1-O(\frac{1}{e^{n^{0.25}}})) \sum_{j=n^{1/2}}^{m} (1-(1-p)^{j})\\
&\geq (1-(1-p)^{n^{1/2}})(1-O(\frac{1}{e^{n^{0.25}}}))\\
&\geq (1-(1-p)^{n^{1/2}})^2 \\
&\geq (1-\frac{2}{x^{O(n)}})
\end{array}
\] 
where $x=(\frac{1}{1-p})>1$. Consequently, with probability at least $(1-\frac{2}{x^{O(n)}})$, there exists no constant which bounds the doubling dimension of $\G{n}{(\log{n})^{\frac{1}{2}-\frac{\theta}{2}}}$.
\end{proof}

\end{document}